\documentclass[12pt]{article}

\usepackage[margin=1in]{geometry}
\usepackage{amsmath,amssymb,amsthm,mathtools,bm}
\usepackage{graphicx,booktabs,tabularx,array}
\usepackage{enumitem}
\usepackage{natbib}
\usepackage{hyperref}
\usepackage[nameinlink,capitalize]{cleveref}

\hypersetup{colorlinks=true,linkcolor=blue,citecolor=blue,urlcolor=blue}
\graphicspath{{figures/}{./}}
\setlist[itemize]{leftmargin=2em}
\setlist[enumerate]{leftmargin=2em}

\newtheorem{assumption}{Assumption}
\newtheorem{definition}{Definition}
\newtheorem{lemma}{Lemma}
\newtheorem{proposition}{Proposition}
\newtheorem{theorem}{Theorem}
\newtheorem{corollary}{Corollary}
\newtheorem{remark}{Remark}

\crefname{assumption}{Assumption}{Assumptions}
\Crefname{assumption}{Assumption}{Assumptions}
\crefname{definition}{Definition}{Definitions}
\Crefname{definition}{Definition}{Definitions}
\crefname{lemma}{Lemma}{Lemmas}
\Crefname{lemma}{Lemma}{Lemmas}
\crefname{proposition}{Proposition}{Propositions}
\Crefname{proposition}{Proposition}{Propositions}
\crefname{theorem}{Theorem}{Theorems}
\Crefname{theorem}{Theorem}{Theorems}
\crefname{corollary}{Corollary}{Corollaries}
\Crefname{corollary}{Corollary}{Corollaries}
\crefname{section}{Section}{Sections}
\Crefname{section}{Section}{Sections}

\newcommand{\Pp}{\mathbb P}
\newcommand{\Ee}{\mathbb E}
\newcommand{\Rr}{\mathbb R}
\newcommand{\ind}{\mathbf{1}}
\newcommand{\calY}{\mathcal Y}
\newcommand{\calX}{\mathcal X}
\newcommand{\calS}{\mathcal S}
\newcommand{\calT}{\mathcal T}
\newcommand{\calA}{\mathcal A}

\newcommand{\Gam}{\Gamma}
\newcommand{\Lam}{\Lambda}
\DeclareMathOperator{\argmin}{arg\,min}
\DeclareMathOperator{\argmax}{arg\,max}

\DeclareMathOperator{\Var}{Var}
\DeclareMathOperator{\Cov}{Cov}

\title{Nested Sensitivity Envelopes for Transported Quantile Treatment Effects}
\author{Pengyun Wang \\ Data Science Institute, The University of Chicago}
\date{Working manuscript \today}

\begin{document}
\maketitle

\begin{abstract}
We study target-population quantile treatment effects when a source study may have unmeasured treatment confounding and may not transport to a target population after conditioning on observed covariates. The observed data consist of a source sample with treatment, outcome and covariates, and a target sample with covariates only. We impose two marginal sensitivity restrictions: an odds-ratio bound \(\Gam\) for source treatment assignment and a conditional likelihood-ratio bound \(\Lam\) for source-to-target potential-outcome distribution shift. For each treatment arm and threshold \(y\), we derive a closed-form sharp target counterfactual CDF envelope.  The envelope nests a source marginal-sensitivity map inside a target outcome-shift map, preserving two normalizations and generally improving on a single product likelihood-ratio relaxation. We prove process-level sharpness, so the envelopes are attainable as entire CDFs and can be inverted to obtain sharp target quantile bounds and sharp interval-hull QTE bounds. We then develop semiparametric theory for these nonsmooth bound processes. On regular index sets, we give the canonical gradient, including the source propensity contribution required in observational studies, and construct cross-fitted Neyman-orthogonal one-step estimators with uniform Gaussian approximation. On full index sets with active-set ties or mass points, we use Hadamard directional differentiability and subsampling-valid inference, with a primitive finite-support route for the required weak convergence. Finally, we invert simultaneous monotone CDF bands to obtain honest confidence sets for quantile and QTE interval-hull processes, and formulate the two-dimensional \((\Gam,\Lam)\) breakdown frontier as level-set inference for interval-hull non-refutation.
\end{abstract}


\section{Introduction}

A common design in causal generalization combines a source study, where treatment and outcome are observed, with a target sample, where only baseline covariates are observed \citep{ColeStuart2010,StuartColeBradshawLeaf2011,DahabrehEtAl2019,PearlBareinboim2011,BareinboimPearl2016,DegtiarRose2023}.  Under source internal validity and conditional transportability of potential outcomes, the target counterfactual distribution and the target quantile treatment effect are identified by standardization to the target covariate distribution \citep{Firpo2007,ChernozhukovFernandezValMelly2013,Wang2026}.  This paper studies the harder case in which both assumptions are allowed to fail in controlled ways.  The source treatment assignment may be confounded by unmeasured variables, and the conditional potential-outcome distribution in the target population may differ from the corresponding source distribution even after conditioning on observed covariates.

The estimand is the target-population quantile treatment effect
\[
        \Delta_T(\tau)=Q^T_1(\tau)-Q^T_0(\tau),\qquad \tau\in(0,1),
\]
where $Q^T_a$ is the quantile function of the target potential outcome $Y^a$ under treatment $a$.  The target sample contains no outcomes, so even standard point identification requires transporting an entire conditional distribution from the source to the target.  Sensitivity analysis for this object is not a direct extension of average-effect sensitivity analysis.  Average effects are linear in the conditional potential-outcome law.  QTEs require first characterizing sharp bounds on counterfactual distribution functions and then passing through a generalized inverse map, which is nonsmooth at mass points, flat regions and envelope contact sets; this feature is central in recent distributional sensitivity analysis under relaxations of unconfoundedness \citep{MastenPoirierRen2025}.

We impose two sensitivity restrictions.  First, the source population satisfies a marginal odds-ratio sensitivity model with parameter $\Gam\ge 1$: the odds of receiving treatment $a$ after conditioning on the potential outcome $Y^a$ and covariates may differ from the observed source propensity odds by at most $\Gam$ \citep{Tan2006,ZhaoSmallBhattacharya2019,DornGuo2023,MastenPoirierRen2025}.  Second, source-to-target outcome shift is controlled by a conditional likelihood-ratio bound $\Lam\ge1$: for each treatment arm and covariate value, the target potential-outcome distribution may be an outcome-dependent tilt of the source potential-outcome distribution, with Radon--Nikodym derivative between $\Lam^{-1}$ and $\Lam$ \citep{AsiaeePalBeckHuling2026}.  When $(\Gam,\Lam)=(1,1)$ the standard point-identified transported QTE model is recovered.

The first contribution is a closed-form sharp distributional envelope under the joint sensitivity model.  Let
\[
        p_a(y,x)=\Pp(Y\le y\mid R=1,A=a,X=x),
        \qquad e_a(x)=\Pp(A=a\mid R=1,X=x),
\]
where $R=1$ denotes membership in the source sample and $R=0$ denotes membership in the target covariate sample.  Define
\[
        \ell_\Gam(e)=e+(1-e)/\Gam,
        \qquad u_\Gam(e)=e+\Gam(1-e),
\]
which are the lower and upper bounds on the inverse treatment-selection tilt induced by the marginal sensitivity model.  The sharp source-population potential-outcome CDF bounds are
\[
 C^-_\Gam(p,e)=\max\{\ell_\Gam(e)p,\,1-u_\Gam(e)(1-p)\},
\quad
 C^+_\Gam(p,e)=\min\{u_\Gam(e)p,\,1-\ell_\Gam(e)(1-p)\}.
\]
The sharp transport-shift CDF maps are
\[
 T^-_\Lam(q)=\max\{\Lam^{-1}q,\,1-\Lam(1-q)\},
\quad
 T^+_\Lam(q)=\min\{\Lam q,\,1-\Lam^{-1}(1-q)\}.
\]
We prove that the sharp target conditional CDF envelope is
\[
        b^-_{a,s}(y,x)=T^-_\Lam\{C^-_\Gam(p_a(y,x),e_a(x))\},
        \qquad
        b^+_{a,s}(y,x)=T^+_\Lam\{C^+_\Gam(p_a(y,x),e_a(x))\},
\]
where $s=(\Gam,\Lam)$.  The target marginal CDF envelope is then
\[
        \psi^-_{a,s}(y)=\Ee\{b^-_{a,s}(y,X)\mid R=0\},
        \qquad
        \psi^+_{a,s}(y)=\Ee\{b^+_{a,s}(y,X)\mid R=0\}.
\]
The proof is constructive: the least-favorable source distribution is obtained by a threshold inverse-selection tilt, and the least-favorable target distribution is obtained by a second threshold outcome-shift tilt.  Thus the bounds are not merely pointwise.  They are attainable as entire CDF paths.  Consequently the sharp target quantile bounds are
\[
        q^-_{a,s}(\tau)=\inf\{y:\psi^+_{a,s}(y)\ge \tau\},
        \qquad
        q^+_{a,s}(\tau)=\inf\{y:\psi^-_{a,s}(y)\ge \tau\},
\]
and the sharp interval-hull bounds for the target QTE are
\[
        \Big[q^-_{1,s}(\tau)-q^+_{0,s}(\tau),\;
             q^+_{1,s}(\tau)-q^-_{0,s}(\tau)\Big].
\]

The second contribution is efficient semiparametric theory for the sharp CDF-bound process, using tangent-space calculations for semiparametric models and orthogonal-score ideas for flexible nuisance estimation \citep{Tsiatis2006,vdV2000,ChernozhukovChetverikovDemirerDufloHansenNeweyRobins2018}.  At active-set regular laws the endpoint map is pathwise differentiable.  If $G^-_s=T^-_\Lam\circ C^-_\Gam$ or $G^+_s=T^+_\Lam\circ C^+_\Gam$ denotes the active endpoint map and $G_p,G_e$ its derivatives with respect to $(p,e)$, the source contribution to the efficient influence function is
\[
        G_p\{p_a(y,X),e_a(X)\}\frac{\ind(A=a)}{e_a(X)}
        \{\ind(Y\le y)-p_a(y,X)\}
        +G_e\{p_a(y,X),e_a(X)\}\{\ind(A=a)-e_a(X)\}.
\]
The second term is essential for observational source studies because the endpoint depends on the source treatment propensity through the sensitivity envelope.  It disappears only when the treatment probability is known by design.  Combining this conditional contribution with the target empirical distribution yields the canonical gradient.  We use it to define cross-fitted one-step estimators and establish uniform asymptotic linearity over regular index sets of thresholds, treatment arms, endpoint signs and sensitivity parameters; full-index nonregular cases are handled by directional inference.

The third contribution is inference that respects the nonsmooth structure of the problem, building on the delta method for quantiles and on modern inference for directionally differentiable maps \citep{vdV2000,FangSantos2019}.  The envelope maps are maxima and minima of affine functions, and the quantile map is a generalized inverse.  Wald intervals for quantile endpoints are therefore fragile and require additional density and active-set assumptions.  Our primary inferential object is a simultaneous confidence band for the CDF-bound process; inversion of this band yields honest confidence sets for quantile bounds and QTE bounds.  Under active-set regularity the critical value can be obtained by a multiplier bootstrap for the efficient process.  Without active-set separation the endpoint maps remain Hadamard directionally differentiable, and we use subsampling to obtain valid critical values for the induced non-Gaussian sup-norm limit.

The fourth contribution is to treat the two-dimensional sensitivity frontier as a statistical object, in the spirit of breakdown-frontier analysis for partially identified causal models \citep{MastenPoirier2020,LannersRudinVolfovskyParikh2025}.  For a fixed quantile index $\tau$, define the lower and upper QTE interval-hull endpoints by $\Delta^-_s(\tau)$ and $\Delta^+_s(\tau)$.  The null value zero lies in the sharp QTE interval hull at sensitivity level $s=(\Gam,\Lam)$ precisely when
\[
        \kappa^{\mathrm{hull}}_s(\tau)=\min\{\Delta^+_s(\tau),-\Delta^-_s(\tau)\}\ge 0.
\]
We construct simultaneous confidence bands for $\kappa^{\mathrm{hull}}_s(\tau)$ over compact rectangles of $(\Gam,\Lam)$ and use them to build inner and outer confidence sets for the interval-hull non-refutation region and its boundary, the breakdown frontier.  This distinction is necessary: with discrete outcomes or support gaps, the exact scalar QTE identified set can be nonconnected even when its sharp interval hull contains zero.

\section{Literature review and gap}\label{sec:lit}

Generalizability and transportability methods combine study data with a target covariate distribution under assumptions that link the conditional potential-outcome law across populations.  Early potential-outcome formulations used standardization, inverse-odds-of-sampling weights, and doubly robust estimators for target average treatment effects \citep{ColeStuart2010,StuartColeBradshawLeaf2011,BuchananEtAl2018,DahabrehEtAl2019}.  The graphical literature formalized related questions under transportability and data fusion \citep{PearlBareinboim2011,BareinboimPearl2016}, and recent reviews emphasize that most mature inferential theory remains centered on first moments \citep{DegtiarRose2023}.  Efficient first-moment transportation has also been studied in detail, including doubly robust and minimax theory for target ATEs \citep{ZengKennedyBodnarNaimi2023}.  Our setting uses the same two-sample structure--a source sample with outcomes and a target sample with covariates only--but targets the entire target counterfactual distribution and its quantile contrasts under partial identification.

Distributional and quantile treatment effects have a separate large literature.  \citet{Firpo2007} develops efficient semiparametric estimation of QTEs under selection on observables, while \citet{ChernozhukovFernandezValMelly2013} provide inference for counterfactual distribution and quantile processes under regression-based models.  \citet{Rothe2010} studies nonparametric distributional policy effects, and empirical work such as \citet{BitlerGelbachHoynes2006} shows why average effects can miss economically relevant heterogeneity; related sorted-effect tools are developed by \citet{Chernozhukov2018}.  The closest recent transport paper is \citet{Wang2026}, who studies efficient transported distributional and quantile treatment effects in a point-identified source-target design with surrogate-assisted missing outcomes.  That work covers canonical gradients and uniform inference for transported distributional/QTE estimands when exchangeability and transportability hold.  It does not address sensitivity envelopes, hidden confounding, source-to-target outcome shift, or sharp partial identification.  Earlier identification results under conditional partial independence already show how distributional treatment-effect bounds can be expressed analytically in single-population settings \citep{MastenPoirier2018a}, and recent likelihood-ratio sensitivity work for linear estimators provides unified calculations for mean-type targets \citep{DornYap2024}.  More generally, \citet{BenMichael2025} develops debiased estimation for bounds defined by conditional linear programs.  That framework is important background for conditional optimization, but it does not deliver the nested closed-form CDF process, the observational-source propensity contribution, or inverse-CDF/frontier inference under simultaneous internal- and external-validity sensitivity.  These results motivate our binary-event calculus while clarifying the remaining gap.

Sensitivity analysis for unmeasured confounding begins from the observation that conditional exchangeability is not testable in observational source studies.  The marginal sensitivity model of \citet{Tan2006} bounds the departure of a latent propensity score from the observed propensity.  \citet{ZhaoSmallBhattacharya2019}, \citet{DornGuo2023}, and \citet{DornGuoKallus2024} develop increasingly sharp and robust inference for average or linear estimands under such models, and \citet{BonviniKennedy2022} study a different sensitivity parameter based on the proportion of confounded units.  Flexible and generalized sensitivity formulations are developed by \citet{FranksDAmourFeller2020} and \citet{FrauenMelnychukFeuerriegel2023}.  The closest distributional identification result is \citet{MastenPoirierRen2025}: they define a broad class of relaxations of unconfoundedness, show the marginal sensitivity model is a special case, and derive sharp bounds for QTEs and DTEs in a single observational population.  Their sharpness theorem for marginal relaxations also clarifies why arm-specific marginal potential-outcome laws can be varied without imposing a cross-arm copula restriction.  Our paper uses this CDF-envelope insight as the internal-validity layer, but the target problem is different: the target sample contributes only covariates, a second likelihood-ratio layer allows source-to-target potential-outcome distribution shift, the estimand is the transported CDF/quantile-bound process, and inference must propagate both sensitivity layers through generalized inverse maps and a two-dimensional frontier.

Sensitivity analysis for external validity is more recent.  Some work parameterizes violations through bias functions or omitted effect modifiers, including \citet{NguyenEtAl2018}, \citet{NieImbensWager2021}, \citet{DahabrehEtAl2023}, and \citet{Huang2024}; another line constrains density ratios between source and target components.  The closest paper is \citet{AsiaeePalBeckHuling2026}, who impose a conditional likelihood-ratio bound between trial and target outcome distributions, derive sharp target ATE bounds, and show a threshold/greedy structure for the least-favorable tilt.  Their paper is explicitly a first-moment generalization analysis from a randomized trial.  It does not consider source hidden confounding, CDF or QTE bounds, nonregular inverse-CDF inference, or a two-dimensional frontier over internal- and external-validity violations.  In contrast, our primitive object is a nested distributional envelope, not a mean endpoint.

Recent partial-identification and data-fusion work is also close in spirit.  \citet{LannersRudinVolfovskyParikh2025} study partial identification when unmeasured confounding and external-validity violations may both occur, and develop doubly robust estimators and breakdown-frontier analysis for mean-type causal effects.  That paper reinforces the value of joint sensitivity parameters, but it does not solve the nonlinear chain
\[
        \text{conditional CDF envelope}\;\longrightarrow\;
        \text{target CDF process}\;\longrightarrow\;
        \text{generalized inverse}\;\longrightarrow\;
        \text{QTE frontier}.
\]
This chain is the main gap addressed here.  It also connects our inferential arguments to the literature on nonsmooth and partially identified parameters, including partial identification in the sense of \citet{Manski2003}, breakdown frontiers \citep{MastenPoirier2020}, directionally differentiable functionals \citep{FangSantos2019}, quantile delta methods \citep{vdV2000}, and empirical-process/subsampling methods \citep{vdVW1996,PolitisRomanoWolf1999}.

\section{Notation and statistical model}\label{sec:notation}

\subsection{Observed data}

Let $O=(R,X,RA,RY)$ be observed, where $R\in\{0,1\}$ indicates data source.  Units with $R=1$ belong to the source study and have observed covariates $X\in\calX$, treatment $A\in\calA=\{0,1\}$ and outcome $Y\in\calY\subset\Rr$.  Units with $R=0$ belong to the target covariate sample and have only $X$ observed.  We regard $A$ and $Y$ as arbitrary placeholders when $R=0$.  The observed data are i.i.d. draws $O_1,\ldots,O_n$ from a law $P$, and $n_r=\sum_{i=1}^n\ind(R_i=r)$ for $r\in\{0,1\}$.

Write
\[
        \pi_r=P(R=r),\qquad P_r(\cdot)=P(\cdot\mid R=r),\qquad r\in\{0,1\}.
\]
Expectations under $P_r$ are denoted by $E_r$.  When the argument is a function of $X$ only, $E_r$ means integration with respect to $P^X_r$.
The conditional covariate laws are $P^X_r$.  When $P^X_0\ll P^X_1$, define the covariate density ratio
\[
        \omega(x)=\frac{dP^X_0}{dP^X_1}(x).
\]
For $a\in\{0,1\}$, define the source treatment propensity and observed source-arm CDF
\[
        e_a(x)=P(A=a\mid R=1,X=x),
        \qquad
        p_a(y,x)=P(Y\le y\mid R=1,A=a,X=x).
\]
The target potential-outcome distribution under treatment $a$ is denoted by $F^T_a(\cdot)$, and its conditional version given $X=x$ by $F^T_a(\cdot\mid x)$.  The target counterfactual CDF and quantile are
\[
        F^T_a(y)=P_0(Y^a\le y),
        \qquad
        Q^T_a(\tau)=\inf\{y:F^T_a(y)\ge \tau\}.
\]
We write $s=(\Gam,\Lam)$ for a pair of sensitivity parameters and let
\[
        \calS=[1,\bar\Gam]\times[1,\bar\Lam]
\]
be a compact sensitivity rectangle used for uniform statements, where $1\le\bar\Gam,\bar\Lam<\infty$.  The symbol $s$ is never used as a source-population indicator; throughout, $R=1$ denotes source and $R=0$ denotes target.  Superscripts $S$ and $T$ refer to source and target potential-outcome laws, respectively.  The quantile index set is
\[
        \calT=[\tau_L,\tau_U]\subset(0,1).
\]
For a covariate-indexed function $f(X)$, $\|f\|_2$ denotes the $L_2(P^X_1)$ norm unless another measure is explicitly stated.  Thus $\|\widehat p_a(y,\cdot)-p_a(y,\cdot)\|_2$, $\|\widehat e_a-e_a\|_2$, and $\|\widehat\omega-\omega\|_2$ are all taken over the source covariate law $P^X_1$.
For each $(s,\tau)$, let $\mathfrak I_{\Delta}(s,\tau)$ denote the exact set of attainable scalar QTE values under the joint sensitivity model.  Its sharp interval hull is denoted by $\operatorname{hull}\mathfrak I_{\Delta}(s,\tau)=[\Delta^-_s(\tau),\Delta^+_s(\tau)]$.  Unless explicitly stated otherwise, frontier statements below concern this sharp interval hull, not exact scalar set membership.

\subsection{Primitive sensitivity maps}

For $e\in(0,1)$ and $\Gam\ge1$ define
\[
        \ell_\Gam(e)=e+(1-e)/\Gam,
        \qquad
        u_\Gam(e)=e+\Gam(1-e).
\]
These are bounds on the inverse treatment-selection likelihood ratio $e/g$, where $g$ is a latent treatment probability.  For $p\in[0,1]$ define
\begin{align*}
        C^-_\Gam(p,e)&=\max\{\ell_\Gam(e)p,\;1-u_\Gam(e)(1-p)\},\\
        C^+_\Gam(p,e)&=\min\{u_\Gam(e)p,\;1-
        \ell_\Gam(e)(1-p)\}.
\end{align*}
For $q\in[0,1]$ and $\Lam\ge1$ define
\begin{align*}
        T^-_\Lam(q)&=\max\{\Lam^{-1}q,\;1-\Lam(1-q)\},\\
        T^+_\Lam(q)&=\min\{\Lam q,\;1-\Lam^{-1}(1-q)\}.
\end{align*}
The nested endpoint maps are
\[
        G^-_s(p,e)=T^-_\Lam\{C^-_\Gam(p,e)\},
        \qquad
        G^+_s(p,e)=T^+_\Lam\{C^+_\Gam(p,e)\}.
\]
The conditional target CDF-bound functions and marginal CDF-bound functionals are
\begin{align*}
        b^-_{a,s}(y,x)&=G^-_s(p_a(y,x),e_a(x)),
        &\psi^-_{a,s}(y)&=E\{b^-_{a,s}(y,X)\mid R=0\},\\
        b^+_{a,s}(y,x)&=G^+_s(p_a(y,x),e_a(x)),
        &\psi^+_{a,s}(y)&=E\{b^+_{a,s}(y,X)\mid R=0\}.
\end{align*}

\subsection{Sensitivity model}

\begin{assumption}[Sampling, consistency, measurability and overlap]\label{ass:basic}
The observations are i.i.d. and $\calX$ is a standard Borel space.  The outcome support is a compact interval $\calY=[\underline y,\bar y]$.  The source sample satisfies consistency, $Y=Y^A$ almost surely conditional on $R=1$.  Regular conditional distributions are taken in jointly Borel versions; in particular $(y,x)\mapsto p_a(y,x)$ is nondecreasing and right-continuous in $y$ and measurable in $(y,x)$.  There exists $\eta_0>0$ such that $\pi_r\in[\eta_0,1-\eta_0]$, $e_a(X)\in[\eta_0,1-\eta_0]$ almost surely conditional on $R=1$, $P^X_0\ll P^X_1$, and $\eta_0\le \omega(X)\le \eta_0^{-1}$ almost surely conditional on $R=1$.  The symbol $\eta_0$ is reserved for overlap constants; nuisance functions are denoted by $\nu$ below.
\end{assumption}

The next restriction is the armwise version of the marginal sensitivity model.  The orientation is chosen so that each treatment arm is analyzed after recoding $A_a=\ind(A=a)$, matching the marginal relaxation logic used for distributional and QTE bounds under unconfoundedness sensitivity \citep{Tan2006,MastenPoirierRen2025}.

\begin{assumption}[Source marginal treatment-confounding sensitivity]\label{ass:gamma}
Fix $\Gam\ge1$.  For each $a\in\{0,1\}$ and $P^X_1$-almost every $x$, there exists a source potential-outcome conditional distribution $F^S_a(\cdot\mid x)$ and a latent arm-$a$ treatment probability
\[
        g_a(y,x)=P(A=a\mid Y^a=y,X=x,R=1)
\]
such that
\[
        \Gam^{-1}\le
        \frac{e_a(x)/(1-e_a(x))}{g_a(y,x)/(1-g_a(y,x))}
        \le \Gam
\]
for $F^S_a(\cdot\mid x)$-almost every $y$, and the observed source-arm law satisfies
\[
        dP_1(Y\in dy\mid A=a,X=x)
        =\frac{g_a(y,x)}{e_a(x)}\,dF^S_a(y\mid x).
\]
The restriction is armwise and marginal: it is imposed on each $Y^a$ separately after recoding $A_a=\ind(A=a)$.  It does not impose a joint sensitivity restriction on $(Y^0,Y^1)$.
\end{assumption}

The second restriction is an external-validity sensitivity model.  It uses a pointwise likelihood-ratio bound on the target-versus-source potential-outcome law, paralleling recent sharp mean-bound work for trial generalization under outcome distribution shift \citep{AsiaeePalBeckHuling2026}.

\begin{assumption}[Conditional source-to-target outcome-shift sensitivity]\label{ass:lambda}
Fix $\Lam\ge1$.  For each $a\in\{0,1\}$ and $P^X_0$-almost every $x$, the target potential-outcome conditional distribution $F^T_a(\cdot\mid x)$ is absolutely continuous with respect to a source potential-outcome law $F^S_a(\cdot\mid x)$ satisfying \Cref{ass:gamma}, and
\[
        \Lam^{-1}\le \frac{dF^T_a(\cdot\mid x)}{dF^S_a(\cdot\mid x)}(y)
        \le \Lam
\]
for $F^S_a(\cdot\mid x)$-almost every $y$.  Because \Cref{ass:basic} imposes $P^X_0\ll P^X_1$, the source conditional laws needed on the target covariate support are defined after choosing $P^X_1$-a.e. versions.
\end{assumption}

\begin{definition}[Joint sensitivity identified set]\label{def:idset}
For $s=(\Gam,\Lam)$, the joint sensitivity identified set for the target arm-$a$ CDF consists of all functions
\[
        y\mapsto \int F^T_a(y\mid x)\,dP^X_0(x)
\]
that can be generated by some collection of conditional distributions and latent probabilities satisfying \Cref{ass:basic,ass:gamma,ass:lambda} and the observed law $P$.
\end{definition}

\begin{lemma}[Armwise odds-ratio orientation and inverse-selection tilts]\label{lem:orientation}
Fix an arm $a$ and covariate value $x$, and write $P^{\mathrm{obs}}_{a,x}$ for the observed conditional law of $Y$ given $(R,A,X)=(1,a,x)$.  Under \Cref{ass:gamma}, the source potential-outcome law can be written as
\[
        dF^S_a(y\mid x)=h_a(y,x)\,dP^{\mathrm{obs}}_{a,x}(y),
        \qquad h_a(y,x)=\frac{e_a(x)}{g_a(y,x)}.
\]
The odds-ratio restriction is equivalent to
\[
        \ell_\Gam\{e_a(x)\}\le h_a(y,x)\le u_\Gam\{e_a(x)\},
        \qquad \int h_a(y,x)\,dP^{\mathrm{obs}}_{a,x}(y)=1.
\]
Conversely, any measurable tilt $h$ satisfying these two displayed conditions defines an admissible source potential-outcome law by $dF^S=h\,dP^{\mathrm{obs}}_{a,x}$ and an admissible latent arm probability by $g=e_a(x)/h$.  Thus the arm-$0$ formulation is the same marginal sensitivity model applied to the recoded treatment $A_0=\ind(A=0)$.

\end{lemma}

\begin{lemma}[Armwise marginal embedding]\label{lem:armembed}
Fix a covariate value $x$ and write $P^{\mathrm{obs}}_{a,x}$ for the observed source-arm law of $Y$ given $(R,A,X)=(1,a,x)$.  For each arm $a\in\{0,1\}$, let $h_a(\cdot,x)$ be a measurable inverse-selection tilt satisfying
\[
        \ell_\Gam\{e_a(x)\}\le h_a(y,x)\le u_\Gam\{e_a(x)\},
        \qquad \int h_a(y,x)\,dP^{\mathrm{obs}}_{a,x}(y)=1,
\]
and define $F^S_a(\cdot\mid x)$ by $dF^S_a=h_a\,dP^{\mathrm{obs}}_{a,x}$.  Then there exists a conditional source law of $(Y^0,Y^1,A)$ given $X=x$ that reproduces the observed source law of $(A,Y)$ given $X=x$, has marginal potential-outcome laws $F^S_0(\cdot\mid x)$ and $F^S_1(\cdot\mid x)$, and satisfies the armwise marginal sensitivity restriction in \Cref{ass:gamma}.  The construction can be made with arbitrary conditional copulas for $(Y^0,Y^1)$ within the strata $A=0$ and $A=1$.

If, in addition, $F^T_0(\cdot\mid x)$ and $F^T_1(\cdot\mid x)$ are any arm-specific target kernels satisfying \Cref{ass:lambda} relative to $F^S_0$ and $F^S_1$, then they can be embedded in a common target potential-outcome law conditional on $X=x$ with an arbitrary copula.  If the input kernels and tilts are measurable in $x$, the embedded source and target laws can be chosen as Markov kernels in $x$.
\end{lemma}

\section{Theory}\label{sec:theory}

This section contains the main identification and inference theory.  The first subsection proves that the joint sensitivity model has a closed-form sharp CDF envelope and that the envelope is sharp as an entire distribution process.  The second subsection derives the semiparametric score theory for estimating that process, separating the smooth active-set case from the nonsmooth case.  The third subsection turns the CDF bands into quantile and QTE inference and then treats the two-dimensional sensitivity frontier as a level-set object.

\subsection{Sharp identification: nested envelopes and process-level sharpness}\label{sec:identification}

The identification argument starts with a binary-event calculation.  This is the algebraic core behind both sensitivity layers.  The event is $\{Y\le y\}$; the tilt is either an inverse treatment-selection tilt or a source-to-target outcome-shift likelihood ratio.  Because the event is binary, the infinite-dimensional conditional optimization collapses to two affine inequalities, a simplification that underlies both marginal sensitivity CDF bounds and likelihood-ratio outcome-shift bounds \citep{MastenPoirierRen2025,AsiaeePalBeckHuling2026}.

\begin{lemma}[Binary-event likelihood-ratio calculus]\label{lem:eventlr}
Let $Z\in\{0,1\}$ with $P(Z=1)=p$.  For constants $0<\ell\le1\le u<\infty$, consider all nonnegative tilts $h$ such that $\ell\le h\le u$ and $E(h)=1$.  Then
\[
        \inf_h E(hZ)=\max\{\ell p,1-u(1-p)\},
        \qquad
        \sup_h E(hZ)=\min\{up,1-\ell(1-p)\}.
\]
Both extrema are attained.  Moreover, as functions of $p$, the lower and upper envelopes are distribution functions whenever $p$ is a distribution function.
\end{lemma}

For QTEs, pointwise sharpness at each threshold is not enough.  The next lemma records the stronger pathwise fact used below: the same least-favorable tilt attains the whole lower or upper CDF envelope at once.  This is the distributional analogue of the threshold solutions used for mean bounds in likelihood-ratio sensitivity problems.

\begin{lemma}[Simultaneous CDF-path attainment under bounded tilts]\label{lem:pathlr}
Let $F$ be a CDF on $\Rr$, and let $0<\ell\le1\le u<\infty$.  Define
\[
        L_{\ell,u}(y)=\max\{\ell F(y),1-u(1-F(y))\},
        \qquad
        U_{\ell,u}(y)=\min\{uF(y),1-\ell(1-F(y))\}.
\]
Then $L_{\ell,u}$ and $U_{\ell,u}$ are CDFs.  There exist probability measures $F_L$ and $F_U$ such that $F_L=L_{\ell,u}$ and $F_U=U_{\ell,u}$, with Radon--Nikodym derivatives $dF_L/dF$ and $dF_U/dF$ lying in $[\ell,u]$ and integrating to one.  The derivatives can be chosen as threshold tilts, with at most one intermediate value on a threshold atom.
\end{lemma}

\begin{lemma}[Measurable least-favorable kernels]\label{lem:measkernel}
Let $x\mapsto F_x$ be a Markov kernel on the compact interval $\calY$, and let $\ell(x),u(x)$ be measurable functions satisfying $0<\ell(x)\le1\le u(x)<\infty$.  Define
\[
        L_x(y)=\max\{\ell(x)F_x(y),1-u(x)(1-F_x(y))\},
        \qquad
        U_x(y)=\min\{u(x)F_x(y),1-\ell(x)(1-F_x(y))\}.
\]
Then $x\mapsto L_x$ and $x\mapsto U_x$ are Markov kernels.  Moreover, they can be generated from $F_x$ by Radon--Nikodym derivatives in $[\ell(x),u(x)]$ that are jointly measurable in $(x,y)$ after the standard threshold-atom convention.
\end{lemma}

The following elementary coherence lemma is used to keep the notation honest.  It confirms that the nested maps are ordered and CDF-preserving, so that integration over the target covariate law produces genuine CDF envelopes rather than merely pointwise numerical bounds.

\begin{lemma}[Coherence of the nested CDF maps]\label{lem:coherence}
For every $s=(\Gam,\Lam)$ and $e\in(0,1)$, the maps $p\mapsto C^-_\Gam(p,e)$, $p\mapsto C^+_\Gam(p,e)$, $q\mapsto T^-_\Lam(q)$ and $q\mapsto T^+_\Lam(q)$ are nondecreasing continuous maps from $[0,1]$ to $[0,1]$, with $C^-_\Gam(0,e)=C^+_\Gam(0,e)=T^-_\Lam(0)=T^+_\Lam(0)=0$ and $C^-_\Gam(1,e)=C^+_\Gam(1,e)=T^-_\Lam(1)=T^+_\Lam(1)=1$.  Moreover
\[
        C^-_\Gam(p,e)\le C^+_\Gam(p,e),
        \qquad T^-_\Lam(q)\le T^+_\Lam(q),
\]
for all $p,q\in[0,1]$.  Consequently, if $y\mapsto p(y)$ is a CDF, then $T^-_\Lam\{C^-_\Gam(p(y),e)\}$ and $T^+_\Lam\{C^+_\Gam(p(y),e)\}$ are CDFs and are ordered.
\end{lemma}

The next lemma applies the binary-event calculation to the source hidden-confounding layer.  It is the distributional analogue of the marginal sensitivity calculations used for average effects in \citet{Tan2006,DornGuo2023}, but here the object is a conditional CDF path.

\begin{lemma}[Sharp source CDF envelope under marginal sensitivity]\label{lem:sourceenv}
Under \Cref{ass:basic,ass:gamma}, for each $a$, $x$ and $y$,
\[
        F^S_a(y\mid x)\in
        \big[C^-_\Gam\{p_a(y,x),e_a(x)\},\;
             C^+_\Gam\{p_a(y,x),e_a(x)\}\big].
\]
The lower and upper functions of $y$ are CDFs and are attainable as entire conditional CDFs by admissible latent treatment propensities.
\end{lemma}

The following lemma is the external-validity analogue.  It mirrors the threshold likelihood-ratio structure in \citet{AsiaeePalBeckHuling2026}, but the target is a binary threshold event rather than a conditional mean.

\begin{lemma}[Sharp target CDF envelope under outcome shift]\label{lem:targetenv}
Let $F$ be any source conditional distribution and let $H$ range over all target conditional distributions satisfying $dH/dF\in[\Lam^{-1},\Lam]$.  Then for every $y$,
\[
        H(y)\in[T^-_\Lam\{F(y)\},T^+_\Lam\{F(y)\}].
\]
The two bounding functions $T^-_\Lam\circ F$ and $T^+_\Lam\circ F$ are CDFs and are attainable by threshold likelihood-ratio tilts with respect to $F$.
\end{lemma}

The two previous lemmas can be composed because the transport map is monotone in the source potential-outcome CDF.  The proposition also records a useful non-collapse fact: the joint model is not equivalent, in general, to a single product likelihood-ratio relaxation relative to the observed source arm.

\begin{proposition}[Nested closed form and non-collapse]\label{prop:noncollapse}
For fixed $(a,x,y)$, the sharp lower and upper bounds on $F^T_a(y\mid x)$ under the joint sensitivity model are
\[
        b^-_{a,s}(y,x)=T^-_\Lam\!\big(C^-_\Gam(p_a(y,x),e_a(x))\big),
        \qquad
        b^+_{a,s}(y,x)=T^+_\Lam\!\big(C^+_\Gam(p_a(y,x),e_a(x))\big).
\]
Let the single product relaxation replace the two normalized tilts by one normalized tilt $k$ relative to the observed source-arm law with bounds
\[
        \ell_\Gam(e_a(x))/\Lam\le k\le u_\Gam(e_a(x))\Lam,
        \qquad E\{k(Y)\mid R=1,A=a,X=x\}=1.
\]
Its CDF bounds are
\[
        B^-_{prod}=\max\{\ell_\Gam(e)p/\Lam,1-u_\Gam(e)\Lam(1-p)\},
        \quad
        B^+_{prod}=\min\{u_\Gam(e)\Lam p,1-\ell_\Gam(e)(1-p)/\Lam\}.
\]
The nested model is never looser: $b^-_{a,s}(y,x)\ge B^-_{prod}$ and $b^+_{a,s}(y,x)\le B^+_{prod}$.  The inequalities can be strict.  For example, with $e_a(x)=0.1$, $\Gam=2$, $\Lam=1.5$ and $p_a(y,x)=0.7$, the nested lower bound is $0.286\overline 6$, while the single product-tilt lower bound is $0.256\overline 6$.
\end{proposition}

The preceding formulas can also be audited by a finite-support convex program.  This is useful because it checks that the closed form solves the actual two-layer constrained problem rather than a heuristic relaxation.

\begin{proposition}[Finite-support linear-program audit]\label{prop:lpcheck}
Fix $(a,x)$ and suppose the observed source-arm distribution has finite support $y_1,\ldots,y_K$ with probabilities $r_j>0$.  Let $E_y=\{j:y_j\le y\}$ and $p=\sum_{j\in E_y}r_j$.  Consider the linear programs
\[
\begin{aligned}
        \underline b(y)&=\min_{q,t} \sum_{j\in E_y}t_j,\qquad
        \overline b(y)=\max_{q,t} \sum_{j\in E_y}t_j,\\
        &\text{subject to } \sum_j q_j=\sum_jt_j=1,\quad
        \ell_\Gam(e)r_j\le q_j\le u_\Gam(e)r_j,\quad
        \Lam^{-1}q_j\le t_j\le \Lam q_j .
\end{aligned}
\]
Then
\[
        \underline b(y)=T^-_\Lam\{C^-_\Gam(p,e)\},
        \qquad
        \overline b(y)=T^+_\Lam\{C^+_\Gam(p,e)\}.
\]
Thus the closed form is exactly the value of the finite-dimensional convex problem induced by the two normalized sensitivity layers.
\end{proposition}

The next theorem is the main identification result.  It strengthens pointwise sharpness to process-level sharpness, which is necessary before the CDF envelope can be inverted to produce sharp quantile bounds.

\begin{theorem}[Sharp target CDF-bound process]\label{thm:cdfsharp}
Suppose \Cref{ass:basic,ass:gamma,ass:lambda} hold for $s=(\Gam,\Lam)$.  Then every target arm-$a$ CDF in the joint sensitivity identified set satisfies
\[
        \psi^-_{a,s}(y)\le F^T_a(y)\le \psi^+_{a,s}(y)
        \qquad\text{for all }y\in\calY.
\]
The functions $\psi^-_{a,s}$ and $\psi^+_{a,s}$ are CDFs on $\calY$, and each is attainable as an entire target counterfactual CDF by a data-generating process satisfying the joint sensitivity model and agreeing with the observed law.  Thus the displayed envelope is process-level sharp.
\end{theorem}

The QTE contrast requires simultaneous armwise attainability.  The next corollary uses \Cref{lem:armembed}, which shows that the armwise marginal sensitivity model imposes restrictions on the two marginal potential-outcome laws but does not impose an additional cross-arm copula restriction.  A joint sensitivity model that constrained $P(A=1\mid Y^0,Y^1,X,R=1)$ would be a different, generally narrower, model; the present paper follows the marginal armwise convention used for QTE bounds under marginal relaxations of unconfoundedness.

The generalized inverse reverses stochastic order: the lower quantile bound is obtained by inverting the upper CDF envelope, and the upper quantile bound by inverting the lower CDF envelope.  The QTE endpoint then combines opposite arm-specific quantile endpoints.

\begin{corollary}[Sharp target quantile and QTE interval hulls]\label{cor:qte}
Under \Cref{ass:basic,ass:gamma,ass:lambda}, the sharp lower and upper endpoints for the target arm-$a$ quantile are
\[
        q^-_{a,s}(\tau)=\inf\{y:\psi^+_{a,s}(y)\ge\tau\},
        \qquad
        q^+_{a,s}(\tau)=\inf\{y:\psi^-_{a,s}(y)\ge\tau\}.
\]
The sharp interval hull of the target QTE identified set is
\[
        \operatorname{hull}\mathfrak I_{\Delta}(s,\tau)
        =\big[\Delta^-_s(\tau),\Delta^+_s(\tau)\big],
\]
where
\[
        \Delta^-_s(\tau)=q^-_{1,s}(\tau)-q^+_{0,s}(\tau),
        \qquad
        \Delta^+_s(\tau)=q^+_{1,s}(\tau)-q^-_{0,s}(\tau).
\]
The two endpoints are sharp: each is attained by an admissible joint law under the armwise model.  If, in addition, the attainable scalar QTE set is connected, then the displayed interval equals the exact scalar identified set.
\end{corollary}

\begin{remark}[Interval hull versus exact scalar set]\label{rem:intervalhull}
The distinction between a sharp interval hull and an exact scalar identified set matters when the outcome support has gaps or mass points.  The CDF envelopes in \Cref{thm:cdfsharp} are process-level sharp and their inverses give non-improvable scalar endpoints.  However, without additional connectedness or continuity conditions, the set of attainable quantile values between the endpoints need not contain every real number in the interval.  For example, suppose there are no covariates, $\Gam=1$, $\Lam=2$, arm 1 has source mass $1/2$ at $0$ and $1/2$ at $2$, and arm 0 is degenerate at $1$.  At $\tau=1/2$, the target mass at $0$ in arm 1 may range from $1/4$ to $3/4$, so the exact arm-1 median set is $\{0,2\}$ and the exact QTE set is $\{-1,1\}$, while the sharp interval hull is $[-1,1]$.  The inferential procedures below cover the sharp endpoint interval.  Statements about exact null inclusion require the connectedness condition stated separately below.
\end{remark}

\begin{assumption}[Optional connectedness for exact scalar QTE sets]\label{ass:connected}
For the sensitivity values and quantile indices under consideration, the attainable scalar QTE set $\mathfrak I_{\Delta}(s,\tau)$ is connected and contains its sharp lower and upper endpoints.  A sufficient primitive condition is that the least-favorable arm-specific CDFs can be connected by admissible mixture paths whose relevant $\tau$-quantiles are unique and vary continuously along the paths, and that the two treatment arms can be varied independently by \Cref{lem:armembed}.
\end{assumption}

\begin{proposition}[A primitive sufficient condition for connectedness]\label{prop:connectedprimitive}
Fix $(s,\tau)$.  For each arm $a$ and $P^X_0$-almost every $x$, suppose there exist two admissible dominated constructions that attain the lower and upper arm-specific quantile endpoints, with source tilts $h_{a,0}(\cdot,x),h_{a,1}(\cdot,x)$ relative to $P^{\mathrm{obs}}_{a,x}$ and transport tilts $r_{a,0}(\cdot,x),r_{a,1}(\cdot,x)$ relative to the corresponding source laws.  For $\rho\in[0,1]$, define
\[
        h_{a,\rho}=\rho h_{a,1}+(1-\rho)h_{a,0},
        \qquad
        r_{a,\rho}=\frac{\rho r_{a,1}h_{a,1}+(1-\rho)r_{a,0}h_{a,0}}
        {\rho h_{a,1}+(1-\rho)h_{a,0}}.
\]
The denominator is strictly positive because $h_{a,j}\ge \ell_\Gam(e_a)>0$.  Assume that the resulting target CDFs $F^T_{a,\rho}$ have unique $\tau$-quantiles and that $\rho\mapsto Q_{a,\rho}(\tau)$ is continuous for $a=0,1$.  Then \Cref{ass:connected} holds at $(s,\tau)$: the exact scalar QTE set is the sharp interval hull $[\Delta^-_s(\tau),\Delta^+_s(\tau)]$.
\end{proposition}

\begin{corollary}[Exact QTE set under connectedness]\label{cor:exactconnected}
Under \Cref{ass:connected}, $\mathfrak I_{\Delta}(s,\tau)=[\Delta^-_s(\tau),\Delta^+_s(\tau)]$.  Consequently exact scalar non-refutation of the null value zero is equivalent to interval-hull non-refutation, $0\in\mathfrak I_{\Delta}(s,\tau)$ if and only if $\Delta^-_s(\tau)\le0\le\Delta^+_s(\tau)$.
\end{corollary}

\begin{corollary}[Reductions]\label{cor:reductions}
If $\Gam=1$, then $C^-_\Gam(p,e)=C^+_\Gam(p,e)=p$, and the bounds reduce to pure outcome-shift transport bounds.  If $\Lam=1$, then $T^-_\Lam(q)=T^+_\Lam(q)=q$, and the bounds reduce to source hidden-confounding CDF bounds standardized to the target covariate distribution.  If $(\Gam,\Lam)=(1,1)$, then $\psi^-_{a,s}=\psi^+_{a,s}=E\{p_a(y,X)\mid R=0\}$, the usual point-identified transported counterfactual CDF.
\end{corollary}

\subsection{Semiparametric estimation of the CDF-bound process}\label{sec:semiparametric}

The sharp envelopes above are functions of the observed conditional CDF $p_a(y,x)$, the source treatment propensity $e_a(x)$, and the target covariate distribution.  This subsection derives the canonical gradient for fixed active sets and then builds a cross-fitted one-step estimator, following the standard semiparametric influence-function calculus \citep{Tsiatis2006,vdV2000}.  The distinction between a randomized source and an observational source enters only through whether $e_a$ is known by design or must be estimated from the source data.

\subsubsection{Active-set notation}

The envelope maps are piecewise affine.  Efficient influence functions are ordinary derivatives only away from ties between active affine pieces.  The following notation records which affine piece is active in each layer.

Let $g^-_{1}(p,e)=\ell_\Gam(e)p$, $g^-_{2}(p,e)=1-u_\Gam(e)(1-p)$, $g^+_{1}(p,e)=u_\Gam(e)p$ and $g^+_{2}(p,e)=1-\ell_\Gam(e)(1-p)$.  Let $t^-_1(q)=\Lam^{-1}q$, $t^-_2(q)=1-\Lam(1-q)$, $t^+_1(q)=\Lam q$ and $t^+_2(q)=1-\Lam^{-1}(1-q)$.  Define active index sets
\begin{align*}
        I^-_C(p,e)&=\argmax_{j\in\{1,2\}}g^-_j(p,e),
        &I^+_C(p,e)&=\argmin_{j\in\{1,2\}}g^+_j(p,e),\\
        I^-_T(q)&=\argmax_{j\in\{1,2\}}t^-_j(q),
        &I^+_T(q)&=\argmin_{j\in\{1,2\}}t^+_j(q).
\end{align*}
At points where both active sets are singletons, $G^-_s$ and $G^+_s$ are differentiable.  We denote their derivatives by
\[
        G^{\sigma}_{p,s}(p,e)=\frac{\partial G^\sigma_s(p,e)}{\partial p},
        \qquad
        G^{\sigma}_{e,s}(p,e)=\frac{\partial G^\sigma_s(p,e)}{\partial e},
        \qquad \sigma\in\{-,+\}.
\]
The derivatives of the primitive pieces are
\begin{align*}
        \ell'_\Gam(e)&=1-\Gam^{-1},        &u'_\Gam(e)&=1-\Gam,\\
        \partial_p g^-_1&=\ell_\Gam(e),     &\partial_e g^-_1&=\ell'_\Gam(e)p,\\
        \partial_p g^-_2&=u_\Gam(e),        &\partial_e g^-_2&=-u'_\Gam(e)(1-p),\\
        \partial_p g^+_1&=u_\Gam(e),        &\partial_e g^+_1&=u'_\Gam(e)p,\\
        \partial_p g^+_2&=\ell_\Gam(e),     &\partial_e g^+_2&=-\ell'_\Gam(e)(1-p),\\
        (t^-_1)'&=\Lam^{-1}, &(t^-_2)'&=\Lam,\\
        (t^+_1)'&=\Lam, &(t^+_2)'&=\Lam^{-1}.
\end{align*}

\begin{assumption}[Active-set regularity on a restricted index set]\label{ass:active}
Let $\mathcal I_{\mathrm{reg}}$ be the index set used for a Gaussian approximation.  For every $(a,\sigma,y,s)\in\mathcal I_{\mathrm{reg}}$, the active sets $I_C^\sigma\{p_a(y,x),e_a(x)\}$ and $I_T^\sigma\{C^\sigma_\Gam(p_a(y,x),e_a(x))\}$ are singletons for $P^X_0$-almost every $x$.  Their active margins are bounded away from zero uniformly over $(a,\sigma,y,s,x)$ on this set.  Equivalently, the following switch surfaces are avoided uniformly:
\[
        p_a(y,x)\ne \frac{\Gam}{\Gam+1}\quad(C^-_\Gam),
        \qquad
        p_a(y,x)\ne \frac{1}{\Gam+1}\quad(C^+_\Gam),
\]
\[
        C^-_\Gam\{p_a(y,x),e_a(x)\}\ne \frac{\Lam}{\Lam+1}\quad(T^-_\Lam),
        \qquad
        C^+_\Gam\{p_a(y,x),e_a(x)\}\ne \frac{1}{\Lam+1}\quad(T^+_\Lam).
\]
Thus active-set Gaussian theory is a regular-index result.  In particular, points with $\Gam=1$ or $\Lam=1$ have tied affine pieces and are handled by the nonsmooth theory unless the maps are simplified before differentiation.  For full CDF-index inference, where a continuous CDF will typically cross one of these surfaces, inference is based on the directional/subsampling route in \Cref{thm:hdd,cor:subsampling}.
\end{assumption}

\subsubsection{Efficient influence functions}

The canonical gradient has two sources of variation: the target empirical distribution of $X$ and the source conditional law used to estimate $p_a$ and, in observational studies, $e_a$.  The source conditional contribution is multiplied by the covariate density ratio so that it is averaged over the target covariate law.

Let $\chi=1$ when $e_a$ is an unknown source propensity and $\chi=0$ when the treatment probability is known by design.  For $\sigma\in\{-,+\}$ define
\begin{align*}
        \zeta^{\sigma}_{a,s,y}(O)
        & =G^{\sigma}_{p,s}\{p_a(y,X),e_a(X)\}
          \frac{\ind(A=a)}{e_a(X)}\{\ind(Y\le y)-p_a(y,X)\}\\
        &\quad +\chi\,
          G^{\sigma}_{e,s}\{p_a(y,X),e_a(X)\}
          \{\ind(A=a)-e_a(X)\}.
\end{align*}
This source contribution is defined only on $R=1$ observations; its conditional mean given $(R=1,X)$ is zero.

\begin{assumption}[Treatment mechanism model]\label{ass:treatmech}
One of the following two models is used.  In the known-design source model, $e_a(x)$ is known by design and $\chi=0$.  In the observational-source model, $e_a(x)$ is an unknown nuisance function estimated from the source sample and $\chi=1$.  No other part of the score changes between the two models.
\end{assumption}

\begin{lemma}[Mixture-sampling tangent decomposition]\label{lem:tangent}
Consider a regular parametric path $P_t$ through $P$ in the observed-data model, with score $S(O)$.  At an active-set regular law, the derivative of $\psi^\sigma_{a,s}(y)$ along the path decomposes into target-covariate, source-outcome, and source-treatment components as
\[
\begin{aligned}
\left.\frac{d}{dt}\psi^\sigma_{a,s,t}(y)\right|_{t=0}
&=E\!\left[\frac{\ind(R=0)}{\pi_0}\{b^\sigma_{a,s}(y,X)-\psi^\sigma_{a,s}(y)\}S(O)\right]\\
&\quad+E\!\left[\frac{\ind(R=1)}{\pi_1}\omega(X)\zeta^\sigma_{a,s,y}(O)S(O)\right].
\end{aligned}
\]
Scores that perturb only the source covariate marginal law or only the sampling probability $\pi_1$ contribute zero because the source residual $\zeta^\sigma_{a,s,y}$ has conditional mean zero given $(R,X)=(1,x)$ and the target term has mean zero given $R=0$.
\end{lemma}

\begin{theorem}[Canonical gradient for regular CDF endpoints]\label{thm:eif}
Suppose \Cref{ass:basic,ass:active,ass:treatmech} holds.  In the nonparametric observed-data model with unknown source conditional outcome law, unknown source propensity when $\chi=1$, and unknown target covariate law, the canonical gradient of $\psi^\sigma_{a,s}(y)$ is
\[
        \phi^{\sigma}_{a,s,y}(O)
        =\frac{\ind(R=0)}{\pi_0}\{b^\sigma_{a,s}(y,X)-\psi^\sigma_{a,s}(y)\}
        +\frac{\ind(R=1)}{\pi_1}\omega(X)\zeta^{\sigma}_{a,s,y}(O).
\]
Consequently the semiparametric efficiency bound for estimating $\psi^\sigma_{a,s}(y)$ is
\[
        \Var\{\phi^{\sigma}_{a,s,y}(O)\}.
\]
\end{theorem}

\begin{proposition}[Neyman-orthogonal score]\label{prop:orthogonal}
Let $\nu$ collect the nuisance functions $(p_a,e_a,\omega)$ and let $\nu^\star$ denote their true values.  The one-step score
\[
        m^{\sigma}_{a,s,y}(O;\psi,\nu)
        =\frac{\ind(R=0)}{\pi_0}\{b^\sigma_{a,s,\nu}(y,X)-\psi\}
        +\frac{\ind(R=1)}{\pi_1}\omega_\nu(X)
          \zeta^{\sigma}_{a,s,y,\nu}(O)
\]
satisfies
\[
        \left.\frac{d}{dt}E\{m^{\sigma}_{a,s,y}(O;\psi^\sigma_{a,s}(y),\nu_t)\}\right|_{t=0}=0
\]
for every regular nuisance path $\nu_t$ through $\nu^\star$, whenever \Cref{ass:active} holds and $\nu_t$ remains in the same active region at $t=0$.
\end{proposition}

\subsubsection{Cross-fitted one-step estimator}

The estimator below is written in one-step form.  Cross-fitting allows flexible nuisance estimation while preserving the orthogonality calculation in \Cref{prop:orthogonal}, as in double/debiased machine learning \citep{ChernozhukovChetverikovDemirerDufloHansenNeweyRobins2018}.  The high-level first-stage assumption is stated in terms of uniform rates and empirical-process control, because the nuisance functions may be estimated by different machine-learning or series methods in applications.  To avoid making the inference theorem depend only on conclusion-like assumptions, \Cref{prop:firststageprimitive} gives a primitive Donsker sufficient condition with all pairwise second-order products made explicit; the high-level version can be read as the machine-learning analogue.

Let the sample be split into $K$ folds.  For an observation $i$ in fold $k$, let $\widehat p^{(-k)}_a$, $\widehat e^{(-k)}_a$ and $\widehat\omega^{(-k)}$ be nuisance estimators trained outside fold $k$.  Let $\widehat G_p^{(-k)}$ and $\widehat G_e^{(-k)}$ be active derivatives evaluated at the estimated nuisances.  Define
\begin{align*}
        \widehat b^{\sigma,(-k)}_{a,s}(y,x)
        &=G^\sigma_s\{\widehat p^{(-k)}_a(y,x),\widehat e^{(-k)}_a(x)\},\\
        \widehat\zeta^{\sigma,(-k)}_{a,s,y}(O_i)
        &=\widehat G^{\sigma,(-k)}_{p,s}(y,X_i)
          \frac{\ind(A_i=a)}{\widehat e^{(-k)}_a(X_i)}
          \{\ind(Y_i\le y)-\widehat p^{(-k)}_a(y,X_i)\}\\
        &\quad +\chi\,
          \widehat G^{\sigma,(-k)}_{e,s}(y,X_i)
          \{\ind(A_i=a)-\widehat e^{(-k)}_a(X_i)\}.
\end{align*}
The estimator is
\begin{align*}
        \widehat\psi^\sigma_{a,s}(y)
        &=\frac{1}{n_0}\sum_{i:R_i=0}
          \widehat b^{\sigma,(-k_i)}_{a,s}(y,X_i)\\
        &\quad +\frac{1}{n_1}\sum_{i:R_i=1}
          \widehat\omega^{(-k_i)}(X_i)
          \widehat\zeta^{\sigma,(-k_i)}_{a,s,y}(O_i),
\end{align*}
where $n_r=\sum_i\ind(R_i=r)$.

\begin{assumption}[First-stage and empirical-process conditions]\label{ass:firststage}
All statements are uniform over the index set used in the theorem.  The nuisance estimators are cross-fitted and uniformly bounded away from the boundary values required by \Cref{ass:basic}.  Write
\[
        \delta_{p,n}=\sup_{a,y}\|\widehat p_a(y,\cdot)-p_a(y,\cdot)\|_2,
        \quad
        \delta_{e,n}=\sup_a\|\widehat e_a-e_a\|_2,
        \quad
        \delta_{\omega,n}=\|\widehat\omega-\omega\|_2,
\]
with $\delta_{e,n}=0$ in the known-design source model.  The second-order drift in the orthogonal expansion is $o_p(n^{-1/2})$; a sufficient pairwise product-rate condition is
\[
        \delta_{p,n}\delta_{\omega,n}
        +\chi\,\delta_{e,n}\delta_{\omega,n}
        +\chi\,\delta_{p,n}\delta_{e,n}
        =o_p(n^{-1/2}).
\]
This condition is deliberately written with the $e$--$\omega$ product: in observational source studies, an error in the source treatment propensity multiplies an error in the covariate density ratio through the source augmentation term.  The cross-fitted empirical-process term generated by replacing $\nu^\star$ by $\widehat\nu$ in the influence function is $o_p(1)$ after multiplication by $\sqrt n$.  On a regular index set, the estimated active sets agree with the true active sets with probability tending to one.  Finally, the influence-function class indexed by $(a,\sigma,y,s)$ is pre-Gaussian and satisfies a Donsker condition or another valid high-dimensional Gaussian approximation.
\end{assumption}

\begin{assumption}[Primitive Donsker sufficient conditions]\label{ass:donsker}
On $\mathcal I_{\mathrm{reg}}$, the nuisance estimators are cross-fitted and take values, with probability tending to one, in fixed uniformly bounded function classes whose induced score class
\[
        \{m^\sigma_{a,s,y}(\cdot;\psi^\sigma_{a,s}(y),\nu): (a,\sigma,y,s)\in\mathcal I_{\mathrm{reg}},
        \nu\in\mathcal N\}
\]
is $P$-Donsker with a square-integrable envelope and is $L_2(P)$-continuous at $\nu^\star$.  The nuisance estimators satisfy
\[
        \sup_{a,y}\|\widehat p_a(y,\cdot)-p_a(y,\cdot)\|_2=o_p(1),\quad
        \sup_a\|\widehat e_a-e_a\|_2=o_p(1),\quad
        \|\widehat\omega-\omega\|_2=o_p(1),
\]
with $\widehat e_a=e_a$ in the known-design model.  Let $\delta_{p,n}$, $\delta_{e,n}$ and $\delta_{\omega,n}$ be as in \Cref{ass:firststage}.  The pairwise product-rate condition
\[
        \delta_{p,n}\delta_{\omega,n}
        +\chi\,\delta_{e,n}\delta_{\omega,n}
        +\chi\,\delta_{p,n}\delta_{e,n}
        =o_p(n^{-1/2})
\]
holds.  The estimated active sets agree with the true active sets on $\mathcal I_{\mathrm{reg}}$ with probability tending to one.
\end{assumption}

\begin{proposition}[Primitive sufficient conditions for the high-level expansion]\label{prop:firststageprimitive}
If \Cref{ass:active,ass:donsker} hold, then \Cref{ass:firststage} holds on $\mathcal I_{\mathrm{reg}}$.
\end{proposition}

\begin{theorem}[Uniform asymptotic linearity and Gaussian approximation]\label{thm:ual}
Suppose \Cref{ass:basic,ass:active,ass:firststage} hold on a compact regular index set
\[
        \mathcal I_{\mathrm{reg}}\subseteq\{(a,\sigma,y,s):a\in\{0,1\},\sigma\in\{-,+\},y\in\calY,s\in\calS\}.
\]
Then, in $\ell^\infty(\mathcal I_{\mathrm{reg}})$,
we have
\[
        \sqrt n\{\widehat\psi^\sigma_{a,s}(y)-\psi^\sigma_{a,s}(y)\}
        =\mathbb G_n\phi^\sigma_{a,s,y}+o_p(1),
\]
where $\mathbb G_n$ is the empirical process.  The process $\mathbb G_n\phi^\sigma_{a,s,y}$ converges weakly to a tight mean-zero Gaussian process $\mathbb Z_\psi$ with covariance function
\[
        \Cov\{\mathbb Z_\psi(i),\mathbb Z_\psi(j)\}
        =E\{\phi_i(O)\phi_j(O)\},
\]
where $i,j\in\mathcal I_{\mathrm{reg}}$ and $\phi_i$ abbreviates the corresponding canonical gradient.
\end{theorem}

\subsubsection{Nonsmooth endpoint process}

Active-set separation is convenient but not fundamental.  At ties, the max/min envelope is not linearly differentiable, yet it is Hadamard directionally differentiable.  This is the correct calculus for nonregular CDF-bound processes and follows the general perspective of \citet{FangSantos2019}; the subsampling implementation below uses the classical $m$-out-of-$n$ logic of \citet{PolitisRomanoWolf1999}.

For directions $(h_p,h_e)$ define the directional derivatives of the source envelope maps by
\begin{align*}
        \dot C^-_{\Gam,p,e}(h_p,h_e)
        &=\max_{j\in I^-_C(p,e)}
          \{\partial_p g^-_j(p,e)h_p+\partial_e g^-_j(p,e)h_e\},\\
        \dot C^+_{\Gam,p,e}(h_p,h_e)
        &=\min_{j\in I^+_C(p,e)}
          \{\partial_p g^+_j(p,e)h_p+\partial_e g^+_j(p,e)h_e\}.
\end{align*}
For a direction $h_q$ define
\begin{align*}
        \dot T^-_{\Lam,q}(h_q)&=\max_{j\in I^-_T(q)} (t^-_j)'(q)h_q,\\
        \dot T^+_{\Lam,q}(h_q)&=\min_{j\in I^+_T(q)} (t^+_j)'(q)h_q.
\end{align*}
The nested directional derivatives are
\[
        \dot G^-_{s,p,e}(h_p,h_e)
        =\dot T^-_{\Lam,C^-_\Gam(p,e)}\{\dot C^-_{\Gam,p,e}(h_p,h_e)\},
\]
\[
        \dot G^+_{s,p,e}(h_p,h_e)
        =\dot T^+_{\Lam,C^+_\Gam(p,e)}\{\dot C^+_{\Gam,p,e}(h_p,h_e)\}.
\]
For the nonsmooth results, let
\[
        \mathcal I=\{(a,\sigma,y,s):a\in\{0,1\},\sigma\in\{-,+\},y\in\calY,s\in\calS\}.
\]

\begin{assumption}[Primitive full-index directional limit]\label{ass:directional}
The estimators of the primitive components $(P^X_0,p_a,e_a)$ admit a tight weak limit in the product space containing $\ell^\infty(\calY\times\calX\times\{0,1\})$ for $p_a$, $\ell^\infty(\calX\times\{0,1\})$ for $e_a$, and bounded signed-measure directions for $P^X_0$.  The primitive paths remain in the domain determined by the overlap and monotonicity restrictions.  For subsampling, the same limit law is consistently reproduced by recomputing the primitive estimators on subsamples.
\end{assumption}

\begin{proposition}[Primitive sufficient condition for full-threshold directional inference]\label{prop:directionalfinite}
Suppose the covariate support is finite, $P^X_r(x)$ is bounded away from zero on this support for $r=0,1$, and $e_a(x)$ is bounded away from zero and one.  Estimate $P^X_0$, $P^X_1$, $e_a(x)$ and $p_a(y,x)$ by the corresponding empirical cell proportions and empirical source-arm CDFs.  Then the primitive process
\[
        \big(\widehat P^X_0-P^X_0,\widehat P^X_1-P^X_1,
        \widehat e_a-e_a,
        \widehat p_a(\cdot,x)-p_a(\cdot,x):a\in\{0,1\},x\in\calX\big)
\]
converges weakly in the product of Euclidean spaces and $\ell^\infty(\calY)$ spaces.  Consequently \Cref{ass:directional} holds on the full index set $\mathcal I$ with $\calS$ compact.  The recomputed $m$-out-of-$n$ subsampling distribution in \Cref{cor:subsampling} is valid on $\mathcal I$ whenever $m\to\infty$, $m/n\to0$, and the limiting sup-norm distribution is continuous at its critical value.  The finite-grid case is an immediate special case.
\end{proposition}

\begin{theorem}[Hadamard directional differentiability]\label{thm:hdd}
Under \Cref{ass:basic,ass:directional}, the map from the primitive components $(P^X_0,p_a,e_a)$ to the CDF-bound process $(\psi^-_{a,s},\psi^+_{a,s})$ is Hadamard directionally differentiable as a map into $\ell^\infty(\mathcal I)$.  Its derivative in direction $(h_0,h_p,h_e)$ is
\[
        \dot\Psi^\sigma_{a,s,y}(h_0,h_p,h_e)
        =h_0\{b^\sigma_{a,s}(y,\cdot)\}
         +E_0\!\big[\dot G^\sigma_{s,p_a(y,X),e_a(X)}
             \{h_p(a,y,X),h_e(a,X)\}\big],
\]
where $h_0$ is a signed-measure direction for the target covariate law.  If a regular estimator of the primitive components has a tight weak limit, then the plug-in CDF-bound estimator has the directional delta-method limit obtained by applying $\dot\Psi$ to that primitive limit.  The limit may be non-Gaussian when active sets have ties.
\end{theorem}

\begin{corollary}[Subsampling validity for nonsmooth CDF bounds]\label{cor:subsampling}
Let $m=m_n\to\infty$ with $m/n\to0$.  Let $\widehat\psi_{m,a,s}^{\sigma,b}(y)$ denote the CDF-bound estimator recomputed on subsample $b$ of size $m$, with the same nuisance-estimation procedure applied inside the subsample.  Suppose the weak limit in \Cref{thm:hdd} has a continuous distribution at the relevant sup-norm critical value.  Then the conditional distribution, over uniformly drawn subsamples, of
\[
        \sqrt m\sup_{(a,\sigma,y,s)\in\mathcal I}
        |\widehat\psi_{m,a,s}^{\sigma,b}(y)-\widehat\psi_{a,s}^{\sigma}(y)|
\]
consistently estimates the distribution of
\[
        \sup_{(a,\sigma,y,s)\in\mathcal I}
        |\dot\Psi^\sigma_{a,s,y}(\mathbb Z_0,\mathbb Z_p,\mathbb Z_e)|.
\]
Thus subsampling critical values give asymptotically valid simultaneous CDF bands without active-set separation.
\end{corollary}

\subsection{Quantile inversion and sensitivity-frontier inference}\label{sec:quantile}

The CDF-bound process is the primary regular object.  Quantile and QTE bounds are obtained by generalized inverse maps.  This subsection therefore uses CDF-band inversion as the default inference method, consistent with the treatment of counterfactual distribution and quantile processes in \citet{ChernozhukovFernandezValMelly2013}.  The optional Wald-type quantile influence functions are stated only under additional smoothness assumptions.

Let $c_{1-\alpha}$ be a critical value satisfying
\[
        \Pp\left(
        \sup_{(a,\sigma,y,s)\in\mathcal I}
        \sqrt n\,|\widehat\psi^\sigma_{a,s}(y)-\psi^\sigma_{a,s}(y)|
        \le c_{1-\alpha}
        \right)\to 1-\alpha.
\]
It may be obtained by multiplier bootstrap under \Cref{ass:active} or by subsampling under the weaker conditions of \Cref{cor:subsampling}.  Define raw truncated CDF bands
\[
        \widehat L^{\sigma,raw}_{a,s}(y)=\max\{0,\widehat\psi^\sigma_{a,s}(y)-c_{1-\alpha}/\sqrt n\},
        \qquad
        \widehat U^{\sigma,raw}_{a,s}(y)=\min\{1,\widehat\psi^\sigma_{a,s}(y)+c_{1-\alpha}/\sqrt n\}.
\]
Because a one-step estimator need not be monotone in $y$, the inverse bands use monotone outer envelopes
\[
        \widehat L^\sigma_{a,s}(y)=\sup_{z\le y}\widehat L^{\sigma,raw}_{a,s}(z),
        \qquad
        \widehat U^\sigma_{a,s}(y)=\inf_{z\ge y}\widehat U^{\sigma,raw}_{a,s}(z).
\]
On the event $\widehat L^{\sigma,raw}_{a,s}\le\psi^\sigma_{a,s}\le\widehat U^{\sigma,raw}_{a,s}$ for all $y$, these monotone envelopes still satisfy $\widehat L^\sigma_{a,s}\le\psi^\sigma_{a,s}\le\widehat U^\sigma_{a,s}$ for all $y$, because each $\psi^\sigma_{a,s}$ is nondecreasing.

For an arbitrary CDF band $L\le F\le U$, define the inverse band operators
\[
        \mathcal Q_L(\tau;L,U)=\inf\{y:U(y)\ge \tau\},
        \qquad
        \mathcal Q_U(\tau;L,U)=\inf\{y:L(y)\ge \tau\}.
\]
The confidence bands for the arm-specific quantile bounds are
\begin{align*}
        \widehat q^{-,lo}_{a,s}(\tau)&=\mathcal Q_L(\tau;\widehat L^+_{a,s},\widehat U^+_{a,s}),
        &\widehat q^{-,hi}_{a,s}(\tau)&=\mathcal Q_U(\tau;\widehat L^+_{a,s},\widehat U^+_{a,s}),\\
        \widehat q^{+,lo}_{a,s}(\tau)&=\mathcal Q_L(\tau;\widehat L^-_{a,s},\widehat U^-_{a,s}),
        &\widehat q^{+,hi}_{a,s}(\tau)&=\mathcal Q_U(\tau;\widehat L^-_{a,s},\widehat U^-_{a,s}).
\end{align*}

\begin{theorem}[Honest quantile-bound inference by CDF-band inversion]\label{thm:quantileband}
If the simultaneous CDF band has asymptotic coverage at least $1-\alpha$, then
\[
        \Pp\left(
        q^-_{a,s}(\tau)\in[\widehat q^{-,lo}_{a,s}(\tau),\widehat q^{-,hi}_{a,s}(\tau)],\;
        q^+_{a,s}(\tau)\in[\widehat q^{+,lo}_{a,s}(\tau),\widehat q^{+,hi}_{a,s}(\tau)]
        \text{ for all }(a,s,\tau)\right)
        \ge 1-\alpha+o(1).
\]
This statement does not require positive densities at the quantiles or uniqueness of the inverse.
\end{theorem}

Define an outer confidence band for the QTE identified set by
\begin{align*}
        \widehat\Delta^{lo}_s(\tau)
        &=\widehat q^{-,lo}_{1,s}(\tau)-\widehat q^{+,hi}_{0,s}(\tau),\\
        \widehat\Delta^{hi}_s(\tau)
        &=\widehat q^{+,hi}_{1,s}(\tau)-\widehat q^{-,lo}_{0,s}(\tau).
\end{align*}

\begin{corollary}[Honest QTE interval-hull confidence band]\label{cor:qteband}
Under the conditions of \Cref{thm:quantileband}, with probability at least $1-\alpha+o(1)$,
\[
        [\Delta^-_s(\tau),\Delta^+_s(\tau)]
        \subseteq
        [\widehat\Delta^{lo}_s(\tau),\widehat\Delta^{hi}_s(\tau)]
\]
for all $(s,\tau)\in\calS\times\calT$.
\end{corollary}

\begin{assumption}[Smooth regular quantiles, optional]\label{ass:smoothq}
For each $(a,\sigma,s,\tau)$, the CDF $\psi^\sigma_{a,s}$ is continuously differentiable in a neighborhood of the relevant quantile and its derivative is bounded away from zero and infinity uniformly over the index set.
\end{assumption}

\begin{corollary}[Regular quantile influence functions]\label{cor:qif}
Under \Cref{ass:smoothq} and active-set regularity, the quantile endpoint maps are pathwise differentiable.  Let $f^\sigma_{a,s}(q)$ denote the derivative of $\psi^\sigma_{a,s}$ at $q$.  Then
\[
        IF\{q^-_{a,s}(\tau)\}
        =-\frac{\phi^+_{a,s,q^-_{a,s}(\tau)}(O)}{f^+_{a,s}(q^-_{a,s}(\tau))},
        \qquad
        IF\{q^+_{a,s}(\tau)\}
        =-\frac{\phi^-_{a,s,q^+_{a,s}(\tau)}(O)}{f^-_{a,s}(q^+_{a,s}(\tau))}.
\]
The corresponding QTE endpoint influence functions are obtained by subtracting the appropriate arm-specific quantile influence functions.
\end{corollary}

\subsubsection{Breakdown frontier inference}\label{sec:frontier}

The final theoretical object is the sensitivity frontier.  It summarizes how much joint departure from internal validity and external validity is needed before the null value can no longer be excluded by the sharp QTE interval hull.  This hull language is deliberate.  As \Cref{rem:intervalhull} shows, exact scalar QTE sets may be nonconnected in discrete or gapped outcome distributions.

For $s\in\calS$ and $\tau\in\calT$, define the interval-hull endpoints
\[
        \Delta^-_s(\tau)=q^-_{1,s}(\tau)-q^+_{0,s}(\tau),
        \qquad
        \Delta^+_s(\tau)=q^+_{1,s}(\tau)-q^-_{0,s}(\tau),
\]
and the signed interval-hull non-refutation function
\[
        \kappa^{\mathrm{hull}}_s(\tau)=\min\{\Delta^+_s(\tau),-\Delta^-_s(\tau)\}.
\]
The null value zero belongs to the sharp QTE interval hull if and only if $\kappa^{\mathrm{hull}}_s(\tau)\ge0$.  The interval-hull non-refutation region and its breakdown frontier are
\[
        \mathcal N^{\mathrm{hull}}(\tau)=\{s\in\calS:\kappa^{\mathrm{hull}}_s(\tau)\ge0\},
        \qquad
        \mathcal F^{\mathrm{hull}}(\tau)=\{s\in\calS:\kappa^{\mathrm{hull}}_s(\tau)=0\}.
\]
Under \Cref{ass:connected}, these coincide with exact scalar non-refutation objects.

Let $\widehat\kappa^{\mathrm{hull}}_s(\tau)$ be the plug-in estimate obtained from the estimated QTE endpoints.  Suppose a critical value $d_{1-\alpha}$ satisfies
\[
        \Pp\left(
        \sup_{(s,\tau)\in\calS\times\calT}
        \sqrt n|\widehat\kappa^{\mathrm{hull}}_s(\tau)-\kappa^{\mathrm{hull}}_s(\tau)|
        \le d_{1-\alpha}
        \right)\to1-\alpha.
\]
Define inner and outer confidence sets
\[
        \widehat{\mathcal N}^{in}_{\alpha}(\tau)
        =\{s:\widehat\kappa^{\mathrm{hull}}_s(\tau)-d_{1-\alpha}/\sqrt n\ge0\},
        \qquad
        \widehat{\mathcal N}^{out}_{\alpha}(\tau)
        =\{s:\widehat\kappa^{\mathrm{hull}}_s(\tau)+d_{1-\alpha}/\sqrt n\ge0\}.
\]

\begin{assumption}[Regular interval-hull non-refutation process]\label{ass:kappa}
On the index set $\calS\times\calT$, \Cref{ass:active,ass:smoothq} hold for all quantile endpoints entering $\Delta^-_s(\tau)$ and $\Delta^+_s(\tau)$.  In addition, either the two arguments of the minimum defining $\kappa^{\mathrm{hull}}_s(\tau)$ are separated uniformly, or inference for $\kappa^{\mathrm{hull}}$ uses the directional/subsampling route described in \Cref{cor:subsampling}.
\end{assumption}

\begin{proposition}[Uniform expansion for the interval-hull non-refutation function]\label{prop:kappa}
Under the smooth separated case of \Cref{ass:kappa},
\[
        \sqrt n\{\widehat\kappa^{\mathrm{hull}}_s(\tau)-\kappa^{\mathrm{hull}}_s(\tau)\}
        =\mathbb G_n\varphi_{\kappa,s,\tau}+o_p(1)
        \quad\text{in }\ell^\infty(\calS\times\calT),
\]
where $\varphi_{\kappa,s,\tau}=\varphi^+_{s,\tau}$ on the region where $\Delta^+_s(\tau)<-\Delta^-_s(\tau)$ and $\varphi_{\kappa,s,\tau}=-\varphi^-_{s,\tau}$ on the region where $-\Delta^-_s(\tau)<\Delta^+_s(\tau)$.  Here
\[
        \varphi^-_{s,\tau}=IF\{q^-_{1,s}(\tau)\}-IF\{q^+_{0,s}(\tau)\},
        \qquad
        \varphi^+_{s,\tau}=IF\{q^+_{1,s}(\tau)\}-IF\{q^-_{0,s}(\tau)\}.
\]
A multiplier bootstrap applied to $\varphi_{\kappa,s,\tau}$ consistently estimates the distribution of $\sup_{s,\tau}|\mathbb G_n\varphi_{\kappa,s,\tau}|$ under the same empirical-process conditions as \Cref{thm:ual}.
\end{proposition}

\begin{theorem}[Sensitivity-region and local frontier inference for the interval hull]
\label{thm:frontier}
If the preceding uniform band for $\kappa^{\mathrm{hull}}$ is valid, then
\[
        \Pp\left(
        \widehat{\mathcal N}^{in}_{\alpha}(\tau)
        \subseteq \mathcal N^{\mathrm{hull}}(\tau)
        \subseteq
        \widehat{\mathcal N}^{out}_{\alpha}(\tau)
        \text{ for all }\tau\in\calT
        \right)
        \ge1-\alpha+o(1).
\]
For the frontier-location statement, let
\[
        \mathcal F^{\mathrm{hull}}_{\mathrm{int}}(\tau)
        =\{s\in\operatorname{int}(\calS):
        \kappa^{\mathrm{hull}}_s(\tau)=0\}.
\]
Let $K\subset\operatorname{int}(\calS)$ be a known compact set.  Assume that $\mathcal F^{\mathrm{hull}}_{\mathrm{int}}(\tau)$ is nonempty and contained in $\operatorname{int}(K)$ uniformly in $\tau$, that $\kappa^{\mathrm{hull}}_s(\tau)$ is continuously differentiable in $s$ on an open neighborhood of $K$, and that
\[
        \inf_{s\in\mathcal F^{\mathrm{hull}}_{\mathrm{int}}(\tau),\tau\in\calT}
        \|\nabla_s\kappa^{\mathrm{hull}}_s(\tau)\|>0.
\]
Assume also that the frontier is isolated inside $K$: for every fixed $\rho>0$ small enough,
\[
        \inf_{\tau\in\calT}\inf\big\{|\kappa^{\mathrm{hull}}_s(\tau)|:
        s\in K,\ \operatorname{dist}(s,\mathcal F^{\mathrm{hull}}_{\mathrm{int}}(\tau))\ge \rho\big\}>0.
\]
Define the local outer zero-level set on $K$ by
\[
        \widehat{\mathcal F}^{out}_{\alpha,K}(\tau)
        =\Big\{s\in K:
        \widehat\kappa^{\mathrm{hull}}_s(\tau)+d_{1-\alpha}/\sqrt n=0\Big\}.
\]
This statement treats $\widehat\kappa^{\mathrm{hull}}$ as a continuous estimator on $K$. In grid implementations, the displayed set is understood after continuous interpolation, or equivalently up to an additional mesh error; the same rate holds when the grid mesh is $o(n^{-1/2})$.
Then
\[
        d_H\{\widehat{\mathcal F}^{out}_{\alpha,K}(\tau),
        \mathcal F^{\mathrm{hull}}_{\mathrm{int}}(\tau)\}=O_p(n^{-1/2})
\]
uniformly in $\tau$, where $d_H$ is Hausdorff distance.  This is a local level-set statement on an interior compact set.  The ordinary topological boundary of $\widehat{\mathcal N}^{out}_{\alpha}(\tau)$ inside the rectangle $\calS$ may contain artificial pieces of $\partial\calS$ and is not claimed to estimate the frontier.  Under \Cref{ass:connected}, the same local statement applies to exact scalar QTE non-refutation.
\end{theorem}

\section{Simulation study}\label{sec:simulation}

The simulations stress-test the four components of the theory rather than tuning an implementation for one favorable data-generating process.  Experiment 1 verifies the closed-form nested envelope against exact finite-dimensional linear programs and illustrates why the nested model is not equivalent to a single product likelihood-ratio relaxation.  Experiment 2 evaluates CDF-band and QTE interval-hull coverage when the data-generating law is constructed to satisfy a known sensitivity pair.  Experiment 3 replaces low-dimensional finite-cell nuisance estimators by high-dimensional machine-learning nuisance fits in an observational-source design.  Experiment 4 targets nonregular features: mass points, active-set switching, subsampling, and two-dimensional frontier inference.  Together the experiments separate identification error, nuisance-estimation error, inverse-map nonregularity, and frontier-level uncertainty.

\subsection{Common implementation}

All experiments use the observed two-sample structure of \Cref{sec:notation}: a source sample of size $n_1$ with $(X,A,Y)$ and a target sample of size $n_0$ with $X$ only.  The reported runs are produced by the accompanying notebook, which calls \texttt{finite\_support\_sharpness\_audit.py}, \texttt{simulation2\_regular\_smooth\_inference.py}, \texttt{simulation3\_hdml\_nuisance\_learning\_best.py}, and \texttt{simulation4\_nonregular\_frontier\_optimized.py}.  The implementation uses $n_{\mathrm{jobs}}=8$ for the large Monte Carlo or finite-LP components.  Exact sample sizes, Monte Carlo replications, resampling counts, grid sizes and oracle Monte Carlo sizes are given in the relevant captions and appendix tables.  This avoids a common ambiguity in simulation sections: the four experiments have different numerical goals and therefore use different grids and resampling budgets.

The main inferential procedure is the monotone CDF-band inversion method from \Cref{sec:quantile}.  We compare it with plug-in hulls without band enlargement, a point-identified transported QTE estimator obtained by setting $(\Gam,\Lam)=(1,1)$, an observational-source ablation that omits $G_e\{\ind(A=a)-e_a(X)\}$, and Wald-type quantile endpoint intervals in the nonregular experiment.  Coverage is simultaneous over the displayed quantile grid.  In Experiments 2--4, standard errors in parentheses are Monte Carlo standard errors for the reported coverage proportions.

\subsection{Experiment 1: finite-support sharpness and non-collapse}

This experiment directly audits \Cref{lem:eventlr,prop:noncollapse,prop:lpcheck,thm:cdfsharp}.  It is an identification-level audit, not an estimation experiment.  We generate independent finite-support conditional cells with outcome support size
\[
        K\in\{2,3,5,8,12,20\},
\]
draw $r=(r_1,\ldots,r_K)\sim \operatorname{Dirichlet}(2,\ldots,2)$, choose threshold events $E_y$, and set $p=\sum_{j\in E_y}r_j$.  For $e\in(0,1)$ and $\Gam,\Lam\in\{1,1.05,1.25,1.5,2,3,5,8\}$, we compute the lower and upper conditional CDF bounds by the closed form
\[
        T^-_{\Lam}\{C^-_{\Gam}(p,e)\},
        \qquad
        T^+_{\Lam}\{C^+_{\Gam}(p,e)\},
\]
and by the exact finite linear programs in \Cref{prop:lpcheck}.  We also compare the nested model with the single product-LR relaxation
\[
        \Lam^{-1}\ell_\Gam(e)r_j\le t_j\le \Lam u_\Gam(e)r_j,
        \qquad \sum_{j=1}^K t_j=1.
\]
Let $W^{\mathrm{nest}}$ and $W^{\mathrm{prod}}$ denote the resulting CDF interval widths.  Since the product-LR relaxation contains every nested feasible distribution but does not enforce the intermediate normalization of $q$, theory predicts $W^{\mathrm{prod}}\ge W^{\mathrm{nest}}$, with strict inequality on nontrivial active regions.  The run uses $6{,}000$ exact finite LP solves, $200{,}000$ vectorized algebraic cases and $20{,}000$ whole-path audits.

\subsection{Experiment 2: regular smooth inference under exact nested tilts}

Experiment 2 evaluates regular CDF-process inference and CDF-band inversion.  It uses finite covariate support $\calX=\{0,\ldots,J-1\}$ with $J=6$.  Source and target covariate masses are multinomial laws proportional to
\[
        \exp\{0.15\sin(2\pi(x+1)/J)-0.10z_x\},
        \qquad
        \exp\{0.70z_x+0.20\cos(2\pi x/J)\},
\]
respectively, where $z_x=\{x-(J-1)/2\}/\{(J-1)/2\}$, with a small uniform mixture added for overlap.  The source treatment mechanism is observational,
\[
        e_1(x)=\operatorname{expit}\{-0.10+0.85z_x+0.25\sin(2\pi x/J)\},
\]
truncated to $[0.18,0.82]$.  Conditional source-arm outcomes follow truncated normals on $[-4.5,4.5]$, with arm-specific means and scales chosen to give smooth regular quantiles.

The true target potential-outcome distribution is generated by exact nested tilts at
\[
        s_0=(\Gam_0,\Lam_0)=(1.60,1.40),
\]
specifically $F_{Y^1\mid X,R=0}=b^-_{1,s_0}$ and $F_{Y^0\mid X,R=0}=b^+_{0,s_0}$.  Thus the target law is admissible at $s_0$ by construction and the QTE is near a boundary of the interval hull.  We compare $s_{\mathrm{low}}=(1.15,1.10)$, $s_0=(1.60,1.40)$ and $s_{\mathrm{high}}=(2.20,1.80)$.  The run uses $B=300$ Monte Carlo replications, $n_1\in\{400,800,1600\}$, $n_0=1.5n_1$, $149$ bootstrap draws, a threshold grid of size $121$ and a dense oracle grid of size $2001$.

\subsection{Experiment 3: high-dimensional nuisance learning with machine learning}\label{sec:sim3_hdml}

Experiment 3 is the machine-learning stress test.  It uses $d=50$, $X\mid R=1\sim N(0,\Sigma)$, $X\mid R=0\sim N(\delta,\Sigma)$, $\Sigma_{jk}=0.35^{|j-k|}$, and
\[
        \delta=(0.85,-0.65,0.75,0.75,0.55,0.45,-0.35,0.35,0.20,0.20,0,\ldots,0)^\top .
\]
The source treatment mechanism is nonlinear and observational:
\[
\begin{aligned}
        e_1(x)&=\operatorname{expit}\{g(x)\},\\
        g(x)&=-0.05+1.05\sin(1.15x_1)-0.90x_2+1.25x_3x_4+0.80\ind(x_3>0,x_4>0)\\
        &\quad -0.75\frac{x_5^2}{1+x_5^2}+0.60\cos(x_6+0.5x_7)+0.55\ind(x_8>0)-0.45\sin(x_9x_{10}),
\end{aligned}
\]
truncated to $[0.05,0.95]$.  Source outcomes are bounded by
\[
        Y=4\tanh\{(m_a(X)+\sigma_a(X)\varepsilon)/4\},
        \qquad \varepsilon\sim N(0,1)\text{ truncated to }[-2.5,2.5].
\]
The target law is generated by exact nested tilts at $s_0=(\Gam_0,\Lam_0)=(8.00,1.10)$, with arm 1 at the lower envelope and arm 0 at the upper envelope.  This makes the point-transport model $(\Gam,\Lam)=(1,1)$ misspecified and creates a propensity-stress design in which the $G_e$ component of the observational-source EIF is active.

The nuisance functions are estimated with cross-fitting.  The propensity is estimated by penalized logistic regression on a deliberately coarse nonlinear basis; the density ratio is estimated by a source--target classifier on a richer basis; and the conditional CDF is estimated separately by arm using a residualized high-dimensional CDF learner with penalized location and scale regressions followed by empirical residual rearrangement.  This construction follows the double/debiased machine-learning logic of \citet{ChernozhukovChetverikovDemirerDufloHansenNeweyRobins2018}.  The optional neural robustness run uses a shallow multilayer perceptron component \citep{GoodfellowBengioCourville2016}, but the reported main run uses the fast penalized/basis learner.  The run uses $B=200$ replications, $n_1\in\{600,1200,2400\}$, $n_0=1.5n_1$, $99$ multiplier bootstrap draws, a threshold grid of size $61$, and an independent oracle sample of size $50{,}000$.

\subsection{Experiment 4: nonregular inversion and two-dimensional frontier inference}

Experiment 4 deliberately violates the regular active-set and smooth-quantile conditions used for Wald inference.  The source-arm outcome distribution is zero-inflated with a continuous truncated-normal component:
\[
        Y=0\text{ with probability }\pi_a(X),
        \qquad
        Y\mid Y\ne0,A=a,X=x,R=1\sim\text{truncated normal on }[-3,3].
\]
The probabilities $\pi_a(X)$ put several target quantiles near the mass point at zero, and $s_0=(\Gam_0,\Lam_0)=(1.20,1.20)$ places active branches near the switch surfaces
\[
        p=\frac{\Gam}{\Gam+1},\qquad p=\frac{1}{\Gam+1},
        \qquad q=\frac{\Lam}{\Lam+1},\qquad q=\frac{1}{\Lam+1}.
\]
We compare recomputed $m$-out-of-$n$ subsampling bands with $m=\lfloor n^{0.6}\rfloor$ and $m=\lfloor n^{0.7}\rfloor$, regular $n$-out-of-$n$ bootstrap bands, and Wald endpoint intervals using local finite-difference density estimates.  The run uses $B=300$ replications, $n_1\in\{500,1000,2000\}$, $n_0=1.5n_1$, $149$ bootstrap draws, $149$ subsampling draws for each exponent, a threshold grid of size $181$, and a $31\times31$ sensitivity grid.  Frontier diagnostics are computed for the interval-hull non-refutation set at $\tau_0=0.5$.

\subsection{Results of simulation studies}

\paragraph{Experiment 1.}
The results in \Cref{fig:sim1_finite_support_audit} support the sharp-identification theory.  Across $6{,}000$ exact finite LPs, the largest lower- and upper-endpoint discrepancies are $4.955\times10^{-11}$ and $3.595\times10^{-11}$, respectively, both below the $10^{-10}$ audit threshold reported in \Cref{tab:sim1_finite_support_audit}.  The product-LR relaxation has zero dominance violations over $200{,}000$ algebraic cases, but it is strictly looser in $16.3\%$ of all cases and $25.83\%$ of the nontrivial subset.  Hence the nested envelope is both numerically sharp and genuinely different from a one-layer product-LR relaxation.

\begin{figure}[t]
    \centering
    {\includegraphics[width=0.98\textwidth]{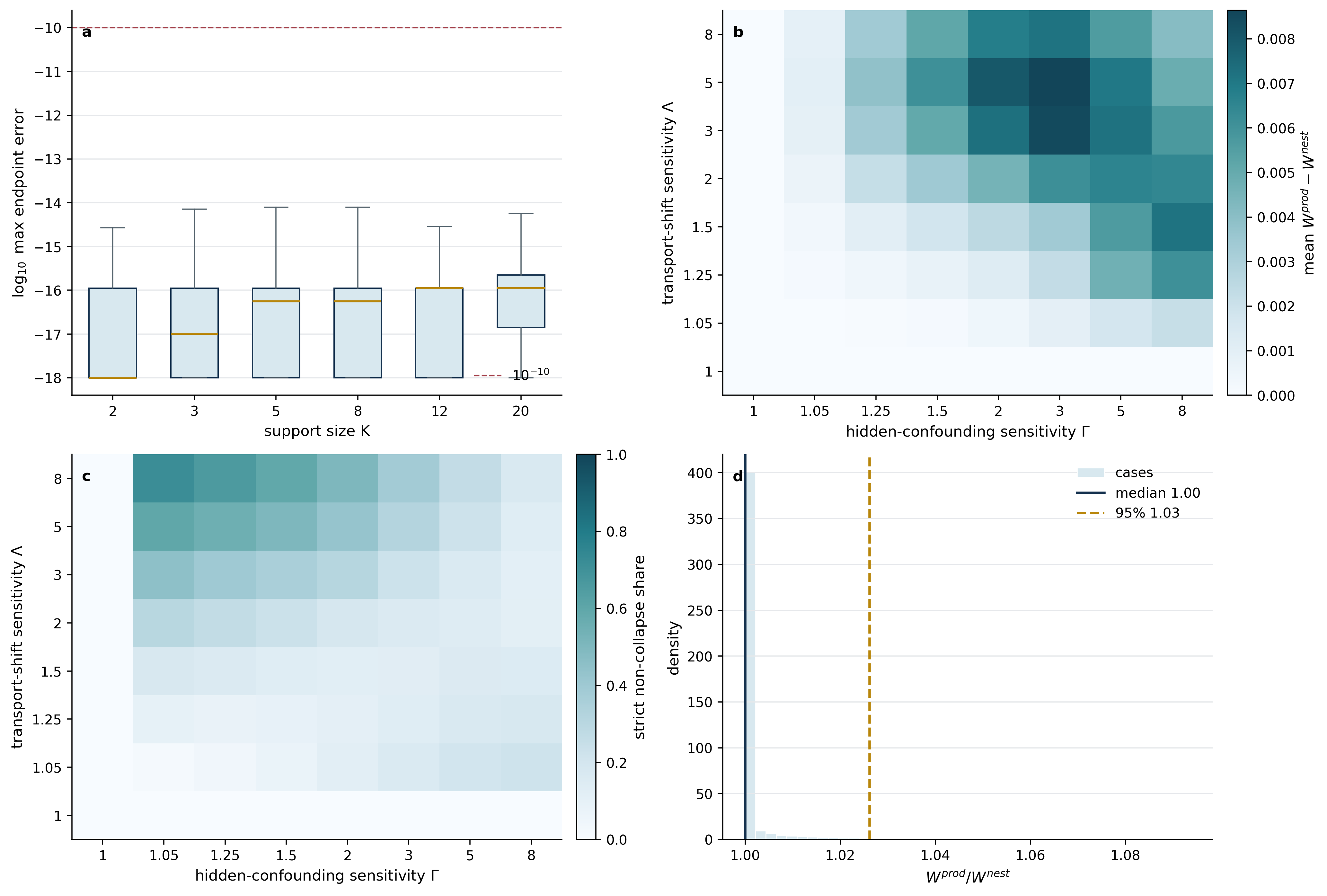}}
    \caption{Finite-support sharpness and non-collapse of the nested sensitivity envelope.  Panel (a) reports the base-10 logarithm of the maximum endpoint discrepancy between the closed-form nested CDF envelope and the exact finite linear programs, stratified by support size $K$.  The dashed line marks the numerical tolerance $10^{-10}$.  Panel (b) reports the average overwidth $W^{\mathrm{prod}}-W^{\mathrm{nest}}$.  Panel (c) reports the share of cases with strict non-collapse.  Panel (d) shows the distribution of $W^{\mathrm{prod}}/W^{\mathrm{nest}}$ over cases with positive nested width.  Full diagnostics are in \Cref{tab:sim1_finite_support_audit}.}
    \label{fig:sim1_finite_support_audit}
\end{figure}

\paragraph{Experiment 2.}
\Cref{tab:sim2_regular_main} separates identification validity from statistical coverage.  For $s_{\mathrm{low}}=(1.15,1.10)$, the true QTE is outside the population interval hull; simultaneous QTE coverage therefore falls to zero at the largest sample size even though the CDF-bound process for that misspecified pair is covered.  For the true and overspecified pairs, the true QTE lies in the interval hull and CDF-band inversion gives conservative simultaneous QTE coverage.  The population hull width increases monotonically from $0.465$ to $1.579$ and $2.689$ as the sensitivity model expands.

\begin{table}[t]
\centering
\setlength{\tabcolsep}{4pt}
\renewcommand{\arraystretch}{1.08}
\caption{Regular smooth inference under an exact nested-tilt DGP at $n_1=1600$ and $n_0=2400$.  The run uses $B=300$ Monte Carlo replications, $149$ bootstrap draws, a threshold grid of size $121$ and a dense oracle grid of size $2001$.  Coverage columns report Monte Carlo proportions with standard errors in parentheses.}
\label{tab:sim2_regular_main}
\begin{tabular}{llccccc}
\toprule
Sensitivity & $(\Gam,\Lam)$ & Truth in hull & QTE cover & CDF cover & Outer width & Pop. width \\
\midrule
underspecified & $(1.15,1.10)$ & no  & $0.000\,(0.000)$ & $0.993\,(0.005)$ & $1.213$ & $0.465$ \\
true & $(1.60,1.40)$ & yes & $1.000\,(0.000)$ & $0.953\,(0.012)$ & $2.602$ & $1.579$ \\
overspecified & $(2.20,1.80)$ & yes & $1.000\,(0.000)$ & $0.953\,(0.012)$ & $3.970$ & $2.689$ \\
\bottomrule
\end{tabular}
\end{table}

\paragraph{Experiment 3.}
\Cref{tab:sim3_hdml_main} shows that the main failure mode in the high-dimensional design is identification.  The point estimator is narrow but invalid: its simultaneous QTE coverage is zero and its endpoint RMSE is $1.247$ at $n_1=2400$.  The full DML procedure remains conservative but honest under high-dimensional nuisance learning.  The plug-in estimator lacks the orthogonal source correction and has much lower QTE coverage.  The No-$G_e$ ablation remains covered because the simultaneous bands are conservative, but omitting $G_e$ increases endpoint RMSE from $0.232$ to $0.286$ and process sup-norm error from $0.218$ to $0.280$.

\begin{table}[t]
\centering
\setlength{\tabcolsep}{4pt}
\renewcommand{\arraystretch}{1.08}
\caption{High-dimensional ML nuisance learning at the largest sample size, $n_1=2400$ and $n_0=3600$.  The run uses $B=200$ Monte Carlo replications, $99$ multiplier bootstrap draws, a threshold grid of size $61$ and an oracle Monte Carlo sample of size $50{,}000$.  Coverage is simultaneous over the quantile grid and uses CDF-band inversion.}
\label{tab:sim3_hdml_main}
\begin{tabular}{lccccc}
\toprule
Method & QTE cover & Hull contain & Endpoint RMSE & Outer width & Proc. err. \\
\midrule
DML & $1.000\,(0.000)$ & $0.470\,(0.035)$ & $0.232$ & $8.335$ & $0.218$ \\
Plug-in & $0.100\,(0.021)$ & $0.065\,(0.017)$ & $0.249$ & $2.343$ & $0.266$ \\
No-$G_e$ & $1.000\,(0.000)$ & $0.390\,(0.034)$ & $0.286$ & $8.507$ & $0.280$ \\
Point & $0.000\,(0.000)$ & $0.000\,(0.000)$ & $1.247$ & $1.177$ & $0.548$ \\
\bottomrule
\end{tabular}
\end{table}

\paragraph{Experiment 4.}
\Cref{tab:sim4_nonregular_main} confirms the nonregular-inference message of the theory.  Wald endpoint intervals are narrower but undercovered: their simultaneous QTE coverage is $0.567$.  CDF-band inversion is conservative and honest.  The subsampling procedures also localize the frontier more accurately than the regular bootstrap: the Hausdorff error is $0.384$ for $m=n^{0.6}$, compared with $0.441$ for the regular bootstrap.  Thus the bootstrap is conservative here, whereas subsampling better matches the local nonregular geometry.

\begin{table}[t]
\centering
\setlength{\tabcolsep}{4pt}
\renewcommand{\arraystretch}{1.08}
\caption{Nonregular inversion and frontier inference at the largest sample size, $n_1=2000$ and $n_0=3000$.  The run uses $B=300$ Monte Carlo replications, $149$ bootstrap draws, $149$ subsampling draws for each exponent, a threshold grid of size $181$ and a $31\times31$ sensitivity grid.}
\label{tab:sim4_nonregular_main}
\begin{tabular}{lccccc}
\toprule
Method & QTE cover & QTE width & Frontier outer & Hausdorff & Plug-in tip err. \\
\midrule
Subsample $m=n^{0.6}$ & $1.000\,(0.000)$ & $1.119$ & $0.977\,(0.009)$ & $0.384$ & $0.046$ \\
Subsample $m=n^{0.7}$ & $1.000\,(0.000)$ & $1.172$ & $0.987\,(0.007)$ & $0.404$ & $0.046$ \\
Regular bootstrap & $1.000\,(0.000)$ & $1.186$ & $0.993\,(0.005)$ & $0.441$ & $0.046$ \\
Wald endpoint & $0.567\,(0.029)$ & $0.681$ & -- & -- & -- \\
\bottomrule
\end{tabular}
\end{table}

\section{Real data illustration}\label{sec:realdata}

We illustrate the proposed procedure using the complete-case NHEFS smoking-cessation data from the \texttt{causaldata} package, which is based on the National Health and Nutrition Examination Survey I Epidemiologic Follow-up Study and is commonly used as a benchmark example in causal inference \citep{CDCNHEFS,HernanRobins2020,CausalData2024}. 
The outcome is 1971--1982 weight change, \(Y=\texttt{wt82\_71}\), and the treatment is smoking cessation, \(A=\texttt{qsmk}\), where \(A=1\) indicates that the subject quit smoking. 
The baseline covariates \(X\) include demographic, smoking-history and baseline health variables available before the post-treatment weight-change outcome: sex, race, income, marital status, education, school years, age, baseline weight \texttt{wt71}, smoking intensity, years smoked, exercise, activity level, hypertension, asthma, bronchitis, diabetes, alcohol frequency, alcohol use in the past year, and baseline price/tax variables when available.

To match the observed-data structure of the paper, we split the public data into a source study and a target covariate sample. 
The source sample consists of younger smokers and has \(A,Y,X\) observed. 
The target sample consists of older smokers and contributes only \(X\); its treatment and outcome columns are masked before estimation and are reported only as design diagnostics in \Cref{tab:realdata_nhefs_appendix}. 
Thus the analysis emulates the sampling structure \(O=(R,X,RA,RY)\): source units have treatment and outcome observed, while target units only contribute covariates. 
The purpose is not to validate coverage, because there is no causal ground truth in the real data. 
Instead, the analysis asks how transported QTE conclusions change when point transportability is replaced by the joint \((\Gam,\Lam)\) sensitivity model.

The estimators use five-fold cross-fitting. 
The source propensity \(e_a(x)\) is estimated from the source sample, the density ratio \(\omega(x)=dP^X_0/dP^X_1\) is estimated from a source--target sampling-score classifier, and the conditional source-arm CDFs \(p_a(y,x)\) are estimated on a threshold grid. 
We report QTE interval hulls at \(\tau\in\{.25,.50,.75\}\). 
The main sensitivity specification is \((\Gam,\Lam)=(2.0,1.5)\), and the table compares it with point transport, one-layer sensitivity analyses, plug-in and No-\(G_e\) ablations, and a product-relaxation benchmark.

\begin{table}[t]
\centering
\setlength{\tabcolsep}{3.5pt}
\renewcommand{\arraystretch}{1.08}
\caption{
NHEFS real-data illustration. 
The source sample consists of younger smokers and the target sample consists of older smokers; target treatments and outcomes are masked in the analysis. 
Entries report estimated transported QTE interval hulls for 1971--1982 weight change. 
They are computed with \(n_1=1009\), \(n_0=413\), 399 multiplier draws for confidence-band diagnostics, 181 threshold grid points, and five-fold cross-fitting.
}
\label{tab:realdata_nhefs_main}
\begin{tabular}{lccccc}
\toprule
Method & \((\Gamma,\Lambda)\) & \(\tau=.25\) & \(\tau=.50\) & \(\tau=.75\) & Avg. width \\
\midrule
Point transport 
& \((1.0,1.0)\) 
& \([3.91, 3.91]\) 
& \([3.58, 3.58]\) 
& \([4.56, 4.56]\) 
& \(0.00\) \\
Hidden-confounding only 
& \((2.0,1.0)\) 
& \([1.95, 7.49]\) 
& \([0.33, 7.82]\) 
& \([0.33, 6.19]\) 
& \(6.30\) \\
Transport-shift only 
& \((1.0,1.5)\) 
& \([0.00, 7.82]\) 
& \([-2.28, 7.82]\) 
& \([0.98, 6.51]\) 
& \(7.82\) \\
Proposed nested DML 
& \((2.0,1.5)\) 
& \([-1.95, 10.42]\) 
& \([-4.89, 12.38]\) 
& \([-1.95, 11.40]\) 
& \(14.33\) \\
Plug-in nested 
& \((2.0,1.5)\) 
& \([-1.95, 11.40]\) 
& \([-3.26, 12.38]\) 
& \([-2.93, 11.73]\) 
& \(14.55\) \\
No-\(G_e\) nested 
& \((2.0,1.5)\) 
& \([-1.95, 10.42]\) 
& \([-4.89, 12.38]\) 
& \([-1.30, 11.40]\) 
& \(14.12\) \\
Product relaxation 
& \((2.0,1.5)\) 
& \([-1.95, 11.40]\) 
& \([-3.26, 12.38]\) 
& \([-2.93, 11.73]\) 
& \(14.55\) \\
\bottomrule
\end{tabular}
\end{table}

\Cref{tab:realdata_nhefs_main} shows how the empirical conclusion changes when internal- and external-validity violations are allowed simultaneously. 
The point-transport analysis, corresponding to \((\Gam,\Lam)=(1,1)\), gives positive transported QTE estimates at all three quantile levels: \(3.91\), \(3.58\), and \(4.56\) pounds. 
Taken literally, this would suggest that smoking cessation increases weight change throughout the target older-smoker distribution. 
The proposed joint sensitivity analysis gives a different message. 
At \((\Gam,\Lam)=(2.0,1.5)\), the nested DML interval hulls are \([-1.95,10.42]\), \([-4.89,12.38]\), and \([-1.95,11.40]\), all of which contain zero. 
Thus the positive point-transport conclusion is not robust once both hidden treatment confounding in the source and source-to-target outcome-distribution shift are allowed.

The one-layer sensitivity rows clarify why the joint model is needed. 
Hidden-confounding-only bounds and transport-shift-only bounds both widen the QTE conclusions, but neither captures the full joint uncertainty. 
For the median QTE, hidden-confounding-only gives \([0.33,7.82]\), transport-shift-only gives \([-2.28,7.82]\), and the proposed nested analysis gives \([-4.89,12.38]\). 
The plug-in and No-\(G_e\) rows are ablations using the same \((\Gam,\Lam)\). 
Their proximity to the DML row shows that the main empirical change is driven by the joint sensitivity model, while their endpoint differences show that the source augmentation and observational-source propensity contribution can affect the estimated bound process. 
In this application, the product relaxation coincides with the plug-in nested estimator at the reported quantile levels, indicating that these quantiles lie in an equality regime; strict non-collapse is a theoretical and simulation-supported possibility, not a feature that must appear at every empirical active region.

Finally, the diagnostics in \Cref{tab:realdata_nhefs_appendix} show that this split is a demanding transport problem. 
The source and target covariate distributions are highly distinguishable, with maximum absolute standardized mean difference \(3.063\) and sampling-score AUC \(0.989\). 
This supports the use of an external-validity sensitivity layer rather than relying only on point transport.

\section{Discussion}\label{sec:discussion}

This paper develops a sharp semiparametric framework for target-population QTE bounds when two identification assumptions can fail simultaneously: source internal validity and conditional source-to-target outcome transportability.  The main identification result is the nested CDF envelope $T_\Lam^\sigma\{C_\Gam^\sigma(p,e)\}$.  The nesting is not a cosmetic composition.  It preserves separate normalizations for treatment-confounding and outcome-transport tilts, which yields bounds that are weakly tighter than a single product likelihood-ratio relaxation and are attainable as whole CDF paths.

The inference theory is deliberately split into regular and nonregular parts.  On active-set regular index sets, the endpoint process is pathwise differentiable and the canonical gradient contains a source-propensity derivative term whenever the source study is observational.  This term is absent only in known-randomization designs.  The high-dimensional simulation confirms that omitting this term degrades centering and process accuracy in a propensity-stress design.  On the full threshold index, active-set ties and mass points are unavoidable in many distributional applications.  For that reason the main inferential recommendation is to construct simultaneous CDF bands and invert them, rather than to rely on density-based Wald intervals for quantile endpoints.

Several limitations are important.  First, the joint sensitivity model is armwise and marginal.  It targets marginal potential-outcome distributions and QTEs; it does not impose a joint latent treatment-assignment model for $(Y^0,Y^1)$, nor does it identify unit-level treatment-effect distributions without additional copula restrictions.  Second, exact scalar QTE identified sets can be nonconnected when outcomes have atoms or support gaps.  The paper therefore formulates the primary QTE and frontier objects as sharp interval hulls.  Exact scalar non-refutation is recovered only under the connectedness condition in \Cref{ass:connected}.  Third, uniform Gaussian approximations require active-set separation and smooth quantile behavior; otherwise the valid route is directional inference and recomputed subsampling.  Fourth, the first-stage assumptions for machine learning are stated through product-rate and weak-convergence conditions.  The simulations show compatibility with high-dimensional nuisance learning, but they are not a universal guarantee for arbitrary black-box learners.

Natural extensions include multi-valued or continuous treatments, longitudinal or multivariate outcomes, covariate-benchmarking tools for calibrating plausible $(\Gam,\Lam)$ regions, and adaptive procedures for choosing regular versus directional inference based on estimated active-set margins.  A local asymptotic minimax theory for the CDF-bound and frontier processes would also clarify which components of interval width are due to irreducible partial identification and which are due to sampling uncertainty.

Overall, the contribution is narrow but hard: simultaneous internal-validity and external-validity sensitivity for transported QTEs requires a nonlinear chain from conditional CDF envelopes to target CDF processes, generalized inverses and two-dimensional frontier level sets.  Existing mean-effect sensitivity and point-identified transported QTE methods do not cover this chain.  Treating the CDF-bound process as the primary object leads to sharp identification and honest inference in both regular and nonregular regimes.

\newpage
\appendix
\setcounter{table}{0}
\renewcommand{\thetable}{A.\arabic{table}}
\renewcommand{\theHtable}{appendix.table.\arabic{table}}
\setcounter{figure}{0}
\renewcommand{\thefigure}{A.\arabic{figure}}
\renewcommand{\theHfigure}{appendix.figure.\arabic{figure}}

\section{Proofs for sharp identification}\label{app:identification}

\begin{proof}[Proof of \Cref{lem:orientation}]
The displayed relationship between the observed source-arm law and the source potential-outcome law gives
\[
        dF^S_a(y\mid x)=\frac{e_a(x)}{g_a(y,x)}\,dP^{\mathrm{obs}}_{a,x}(y),
\]
so $h_a=e_a/g_a$ and $\int h_a\,dP^{\mathrm{obs}}_{a,x}=1$ because $F^S_a(\cdot\mid x)$ is a probability measure.  Solving the odds-ratio inequalities in \Cref{ass:gamma} for $h=e/g$ gives
\[
        e+(1-e)/\Gam\le h\le e+\Gam(1-e),
\]
where $e=e_a(x)$; these are exactly $\ell_\Gam(e)$ and $u_\Gam(e)$.  Conversely, suppose $h$ satisfies the two displayed conditions and set $dF^S=h\,dP^{\mathrm{obs}}_{a,x}$ and $g=e/h$.  Since $h\ge \ell_\Gam(e)> e$ for finite $\Gam\ge1$ and $e<1$, and since $h<\infty$, $g\in(0,1)$; the lower overlap makes the odds well defined.  The same algebra gives the required odds-ratio bound.  Finally,
\[
        \frac{g(y,x)}{e}\,dF^S(y)=\frac{e/h(y,x)}{e}h(y,x)dP^{\mathrm{obs}}_{a,x}(y)=dP^{\mathrm{obs}}_{a,x}(y),
\]
so the observed law is reproduced.  Replacing the binary treatment by $A_a=\ind(A=a)$ proves the arm-$0$ orientation statement.
\end{proof}

\begin{proof}[Proof of \Cref{lem:armembed}]
Fix $x$ and abbreviate $e_a=e_a(x)$ and $P^{\mathrm{obs}}_a=P^{\mathrm{obs}}_{a,x}$.  For each arm define the counter-arm conditional law
\[
        \widetilde F^{1-a}_a(dy\mid x)
        =\frac{F^S_a(dy\mid x)-e_aP^{\mathrm{obs}}_a(dy)}{1-e_a}
        =\frac{h_a(y,x)-e_a}{1-e_a}\,P^{\mathrm{obs}}_a(dy).
\]
This is a probability measure.  Indeed, $h_a\ge\ell_\Gam(e_a)\ge e_a$, so the numerator is nonnegative, and its total mass is $(1-e_a)^{-1}\{1-e_a\}=1$.  Now draw $A$ with $P(A=a\mid X=x)=e_a$.  Conditional on $A=a$, set $Y^a$ to have law $P^{\mathrm{obs}}_a$ and set $Y^{1-a}$ to have law $\widetilde F^{a}_{1-a}$; choose any copula between these two conditional marginals.  This construction gives $Y=Y^A$ and reproduces the observed source-arm laws.

The marginal law of $Y^a$ is
\[
        e_aP^{\mathrm{obs}}_a+(1-e_a)\widetilde F^{1-a}_a=F^S_a.
\]
Bayes' rule gives, for $F^S_a$-almost every $y$,
\[
        P(A=a\mid Y^a=y,X=x,R=1)
        =e_a\frac{dP^{\mathrm{obs}}_a}{dF^S_a}(y)=\frac{e_a}{h_a(y,x)},
\]
which satisfies \Cref{ass:gamma} by \Cref{lem:orientation}.  This proves source embedding.  For the target population, no treatment assignment is observed or restricted; any two marginal kernels $F^T_0(\cdot\mid x)$ and $F^T_1(\cdot\mid x)$ can be coupled by a product copula, or by any other measurable copula, to form a joint conditional law of $(Y^0,Y^1)$ given $X=x$.  If the input kernels are measurable in $x$, the displayed counter-arm kernels are measurable by the Radon--Nikodym construction, and the Ionescu--Tulcea extension theorem yields measurable source and target Markov kernels.
\end{proof}

\begin{proof}[Proof of \Cref{lem:eventlr}]
The cases $p=0$ and $p=1$ are immediate, so assume $p\ne0,1$.  For any feasible $h$,
\begin{align*}
        E(hZ)&\ge \ell P(Z=1)=\ell p,\\
        1-E(hZ)&=E\{h(1-Z)\}\le uP(Z=0)=u(1-p),
\end{align*}
which gives $E(hZ)\ge \max\{\ell p,1-u(1-p)\}$.  Similarly,
\begin{align*}
        E(hZ)&\le up,\\
        1-E(hZ)&=E\{h(1-Z)\}\ge \ell(1-p),
\end{align*}
which gives $E(hZ)\le \min\{up,1-\ell(1-p)\}$.

For attainability of the lower endpoint, write $m=E(hZ)$.  Feasibility is equivalent to choosing $m\in[\ell p,up]$ and $1-m\in[\ell(1-p),u(1-p)]$.  Therefore the feasible interval for $m$ is
\[
        [\ell p,up]\cap[1-u(1-p),1-\ell(1-p)],
\]
whose left endpoint is $\max\{\ell p,1-u(1-p)\}$.  Choose $h$ to be constant on the two atoms $Z=1$ and $Z=0$ with conditional means $m/p$ and $(1-m)/(1-p)$.  These constants lie in $[\ell,u]$ by construction, so the lower endpoint is attained.  The upper endpoint is the right endpoint of the same feasible interval and is attained by the same construction.  For non-atomic or mixed distributions, the same conditional-mean construction can be implemented by constants on the sets $\{Z=1\}$ and $\{Z=0\}$; no additional randomization is needed.

If $p=F(y)$ is a CDF, then $y\mapsto\max\{\ell F(y),1-u(1-F(y))\}$ is nondecreasing and right-continuous because it is the maximum of two nondecreasing right-continuous functions.  Its limits at $-\infty$ and $+\infty$ are $0$ and $1$.  The upper envelope is the minimum of two nondecreasing right-continuous functions and has the same endpoint limits.  Thus both envelopes are CDFs.
\end{proof}

\begin{proof}[Proof of \Cref{lem:pathlr}]
The functions $L_{\ell,u}$ and $U_{\ell,u}$ are CDFs because they are maxima or minima of nondecreasing right-continuous functions with limits zero and one.  By \Cref{lem:eventlr}, $L_{\ell,u}(y)$ and $U_{\ell,u}(y)$ are the pointwise extrema of $E(h\ind(Y\le y))$ over all $h\in[\ell,u]$ with $E(h)=1$.  It remains to show that one tilt attains all thresholds simultaneously.  If $\ell=u=1$, then the only feasible tilt is $h\equiv1$ and both envelopes equal $F$, so the claim is immediate.  Hence suppose $\ell<u$.

For the lower envelope, define the threshold probability $r_L=(u-1)/(u-\ell)\in[0,1]$ and let $t_L=\inf\{y:F(y)\ge r_L\}$.  Set $h_L(y)=\ell$ for $y<t_L$ and $h_L(y)=u$ for $y>t_L$.  If $F$ has an atom at $t_L$, write $\Delta_L=F(t_L)-F(t_L-)$.  When $\Delta_L>0$, set
\[
        h_L(t_L)=\frac{1-\ell F(t_L-)-u\{1-F(t_L)\}}{\Delta_L};
\]
the inequalities $F(t_L-)\le r_L\le F(t_L)$ imply $h_L(t_L)\in[\ell,u]$.  If $\Delta_L=0$, the value at $t_L$ is immaterial.  This choice gives $\int h_L\,dF=1$.  Direct calculation gives
\[
        \int_{(-\infty,y]} h_L\,dF
        =\max\{\ell F(y),1-u(1-F(y))\}=L_{\ell,u}(y)
\]
for $y<t_L$, $y=t_L$ and $y>t_L$ separately.  Thus $F_L(( -\infty,y])=L_{\ell,u}(y)$ with $dF_L/dF=h_L\in[\ell,u]$.  The upper envelope is analogous, using the threshold probability $r_U=(1-\ell)/(u-\ell)$ and assigning $u$ below the threshold and $\ell$ above it.  If $t_U=F^{-1}(r_U)$ and $\Delta_U=F(t_U)-F(t_U-)>0$, the threshold value is
\[
        h_U(t_U)=\frac{1-uF(t_U-)-\ell\{1-F(t_U)\}}{\Delta_U},
\]
with arbitrary value on a zero jump.  This yields $dF_U/dF\in[\ell,u]$ and CDF $U_{\ell,u}$.
\end{proof}

\begin{proof}[Proof of \Cref{lem:measkernel}]
For each fixed $y$, measurability of $x\mapsto F_x(y)$ and of $\ell(x),u(x)$ implies measurability of $x\mapsto L_x(y)$ and $x\mapsto U_x(y)$.  By \Cref{lem:pathlr}, for each $x$ these functions are CDFs on the compact interval $\calY$.  A family of CDFs whose evaluations at every real $y$ are measurable defines a Markov kernel: first define the kernel on intervals $(-\infty,y]$ by the displayed CDF values, then extend to Borel sets by the monotone-class theorem.

On the set where $\ell(x)=u(x)=1$, take the derivative identically equal to one.  On the complementary set, the threshold probabilities $r_L(x)=\{u(x)-1\}/\{u(x)-\ell(x)\}$ and $r_U(x)=\{1-\ell(x)\}/\{u(x)-\ell(x)\}$ are measurable.  Generalized quantiles $t_L(x)=\inf\{y:F_x(y)\ge r_L(x)\}$ and $t_U(x)=\inf\{y:F_x(y)\ge r_U(x)\}$ are measurable because $F_x(y)$ is a measurable kernel CDF.  The intermediate values on threshold atoms are the explicit ratios displayed in the proof of \Cref{lem:pathlr}, with $F,\ell,u,t$ replaced by $F_x,\ell(x),u(x),t(x)$.  The numerator and denominator are measurable functions of $x$ because $x\mapsto F_x(t(x))$ and $x\mapsto F_x(t(x)-)$ are measurable for a measurable quantile $t(x)$; on zero jumps we set any measurable value in $[\ell(x),u(x)]$, for instance one.  Therefore the threshold derivatives constructed in the proof of \Cref{lem:pathlr} are jointly measurable in $(x,y)$ and generate the kernels $L_x$ and $U_x$.
\end{proof}

\begin{proof}[Proof of \Cref{lem:coherence}]
Each displayed map is the maximum or minimum of two affine continuous functions on $[0,1]$, hence is continuous and nondecreasing.  The endpoint values follow by substitution.  The inequalities $C^-_\Gam\le C^+_\Gam$ and $T^-_\Lam\le T^+_\Lam$ are equivalent to nonemptiness of the binary-event feasible intervals in \Cref{lem:eventlr}; explicitly, $0<\ell_\Gam(e)\le1\le u_\Gam(e)$ and $0<\Lam^{-1}\le1\le\Lam$, so the feasible intervals are nonempty for every event probability.  A nondecreasing continuous map preserving zero and one sends CDFs to CDFs; composing such maps preserves this property and the order.
\end{proof}

\begin{proof}[Proof of \Cref{lem:sourceenv}]
Fix $(a,x)$ and write $P^{\mathrm{obs}}_{a,x}$ for the observed source-arm law.  By \Cref{lem:orientation}, admissible source potential-outcome laws are precisely the tilted laws $h\,dP^{\mathrm{obs}}_{a,x}$ with $\ell_\Gam(e_a(x))\le h\le u_\Gam(e_a(x))$ and $\int h\,dP^{\mathrm{obs}}_{a,x}=1$.  Applying \Cref{lem:eventlr} to the event $\{Y\le y\}$, whose observed probability is $p_a(y,x)$, gives the displayed lower and upper bounds.  Applying \Cref{lem:pathlr} and \Cref{lem:measkernel} to the whole observed source-arm CDF path gives simultaneously attainable lower and upper CDF kernels with measurable threshold tilts.  The converse construction in \Cref{lem:orientation} turns those tilts into admissible latent arm probabilities, completing sharpness.
\end{proof}

\begin{proof}[Proof of \Cref{lem:targetenv}]
Fix $(a,x)$ and an admissible source potential-outcome law $F^S_a(\cdot\mid x)$.  Under \Cref{ass:lambda}, admissible target laws are exactly tilts $r\,dF^S_a$ with $\Lam^{-1}\le r\le\Lam$ and $\int r\,dF^S_a=1$.  \Cref{lem:eventlr} applied to $\{Y^a\le y\}$ gives $T^-_\Lam(F^S_a(y\mid x))$ and $T^+_\Lam(F^S_a(y\mid x))$.  Since $T^-_\Lam$ and $T^+_\Lam$ are nondecreasing by \Cref{lem:coherence}, minimizing or maximizing over the sharp source envelope from \Cref{lem:sourceenv} gives the nested endpoints.  \Cref{lem:pathlr,lem:measkernel} give measurable threshold outcome-shift kernels that attain the whole lower or upper target CDF path.
\end{proof}

\begin{proof}[Proof of \Cref{prop:noncollapse}]
By \Cref{lem:sourceenv}, the feasible values of the source potential-outcome CDF at $(a,x,y)$ form the interval
\[
        [C^-_\Gam(p_a(y,x),e_a(x)),C^+_\Gam(p_a(y,x),e_a(x))].
\]
By \Cref{lem:targetenv}, conditional on any feasible source value $q$, the target value lies in $[T^-_\Lam(q),T^+_\Lam(q)]$.  The maps $T^-_\Lam$ and $T^+_\Lam$ are nondecreasing in $q$.  Therefore the smallest feasible target value is obtained at the smallest feasible source value, and the largest feasible target value is obtained at the largest feasible source value.  This gives the displayed nested formula.  Attainability follows by composing the source threshold tilt from \Cref{lem:sourceenv} with the target threshold tilt from \Cref{lem:targetenv}.

For non-collapse, let $p=p_a(y,x)$ and $e=e_a(x)$.  The single product-tilt lower bound with product bounds $[\ell_\Gam(e)/\Lam,u_\Gam(e)\Lam]$ is
\[
        B^-_{prod}=\max\{\ell_\Gam(e)p/\Lam,1-u_\Gam(e)\Lam(1-p)\}.
\]
Because $C^-_\Gam(p,e)\ge \ell_\Gam(e)p$ and $C^-_\Gam(p,e)\ge 1-u_\Gam(e)(1-p)$,
\[
        T^-_\Lam\{C^-_\Gam(p,e)\}
        =\max\{C^-_\Gam(p,e)/\Lam,1-\Lam[1-C^-_\Gam(p,e)]\}
        \ge B^-_{prod}.
\]
Similarly, because $C^+_\Gam(p,e)\le u_\Gam(e)p$ and $C^+_\Gam(p,e)\le 1-\ell_\Gam(e)(1-p)$,
\[
        T^+_\Lam\{C^+_\Gam(p,e)\}
        =\min\{\Lam C^+_\Gam(p,e),1-[1-C^+_\Gam(p,e)]/\Lam\}
        \le B^+_{prod}.
\]
With $e=0.1$, $\Gam=2$, $\Lam=1.5$ and $p=0.7$, $\ell_\Gam(e)=0.55$ and $u_\Gam(e)=1.9$.  The source lower value is
\[
        C^-_\Gam(0.7,0.1)=\max\{0.55\cdot0.7,1-1.9\cdot0.3\}=0.43.
\]
The nested lower target value is
\[
        T^-_{1.5}(0.43)=\max\{0.43/1.5,1-1.5(1-0.43)\}=0.286\overline6.
\]
The product-tilt lower value is
\[
        \max\{(0.55/1.5)0.7,1-(1.9)(1.5)(0.3)\}=0.256\overline6.
\]
Thus the product relaxation is strictly looser in this admissible configuration.
\end{proof}

\begin{proof}[Proof of \Cref{prop:lpcheck}]
The variables $q=(q_1,\ldots,q_K)$ are the source potential-outcome probabilities relative to the observed source-arm support.  The constraints $\ell_\Gam(e)r_j\le q_j\le u_\Gam(e)r_j$ and $\sum_jq_j=1$ are exactly the finite-support version of the source inverse-selection tilt.  By \Cref{lem:eventlr}, the feasible values of $q(E_y)=\sum_{j\in E_y}q_j$ are the interval $[C^-_\Gam(p,e),C^+_\Gam(p,e)]$.

Conditional on any such $q$, the variables $t$ are target probabilities satisfying a likelihood-ratio bound relative to $q$.  Applying \Cref{lem:eventlr} again gives the feasible interval for $t(E_y)$ as $[T^-_\Lam\{q(E_y)\},T^+_\Lam\{q(E_y)\}]$.  Since $T^-_\Lam$ and $T^+_\Lam$ are nondecreasing, minimizing and maximizing over the feasible interval for $q(E_y)$ gives the stated values.  This proves that the closed form is the exact value of the finite linear programs.
\end{proof}

\begin{proof}[Proof of \Cref{thm:cdfsharp}]
Validity is immediate from \Cref{lem:sourceenv,lem:targetenv}: for each $x$, every admissible target conditional CDF lies between $b^-_{a,s}(\cdot,x)$ and $b^+_{a,s}(\cdot,x)$, and integration over $P^X_0$ preserves the order.  By \Cref{lem:coherence}, the two integrated functions are ordered CDFs.

For sharpness of the lower process, \Cref{lem:sourceenv} supplies a measurable source least-favorable kernel attaining $C^-_\Gam\{p_a(\cdot,x),e_a(x)\}$ for $P^X_0$-almost every $x$, and \Cref{lem:targetenv} supplies a measurable target outcome-shift kernel attaining $T^-_\Lam$ of that source kernel.  \Cref{lem:measkernel} ensures that the threshold locations and atom-splitting constants can be chosen as measurable functions of $x$.  Combining these conditional kernels with the target covariate law $P^X_0$ yields an admissible target potential-outcome law whose marginal CDF is exactly $\psi^-_{a,s}$.  The upper process is identical with the upper threshold kernels.  Therefore no uniformly tighter lower or upper CDF process can be valid over the joint sensitivity model.
\end{proof}

\begin{proof}[Proof of \Cref{cor:qte}]
If a CDF $F$ is known only to satisfy $L\le F\le U$, then its left quantile satisfies
\[
        \inf\{y:U(y)\ge\tau\}\le F^{-1}(\tau)\le \inf\{y:L(y)\ge\tau\}.
\]
Apply this with $L=\psi^-_{a,s}$ and $U=\psi^+_{a,s}$.  \Cref{thm:cdfsharp} shows that the two endpoint CDFs are attainable for each arm as entire processes.  \Cref{lem:armembed} allows the arm-1 and arm-0 endpoint-achieving kernels needed for $\Delta^-_s$ and $\Delta^+_s$ to be embedded in a common joint potential-outcome law.  Hence the displayed QTE endpoints are attainable and all attainable scalar QTE values lie in the interval between them, proving sharpness of the interval hull.
\end{proof}

\begin{proof}[Proof of \Cref{prop:connectedprimitive}]
Convexity of the source sensitivity constraint gives
\[
        \ell_\Gam(e_a)\le h_{a,\rho}\le u_\Gam(e_a),\qquad \int h_{a,\rho}\,dP^{\mathrm{obs}}_{a,x}=1.
\]
The displayed $r_{a,\rho}$ is a weighted average of $r_{a,0}$ and $r_{a,1}$ with nonnegative weights proportional to $\rho h_{a,1}$ and $(1-\rho)h_{a,0}$.  Hence $r_{a,\rho}\in[\Lam^{-1},\Lam]$.  Moreover,
\[
        \int r_{a,\rho} h_{a,\rho}\,dP^{\mathrm{obs}}_{a,x}
        =\rho\int r_{a,1}h_{a,1}\,dP^{\mathrm{obs}}_{a,x}
         +(1-\rho)\int r_{a,0}h_{a,0}\,dP^{\mathrm{obs}}_{a,x}=1,
\]
so the interpolated target law is admissible under the two-layer sensitivity model.  By assumption, $\rho\mapsto Q_{a,\rho}(\tau)$ is continuous.  By \Cref{lem:armembed}, the two arms may be varied independently, so the map
\[
        (\rho_1,\rho_0)\mapsto Q_{1,\rho_1}(\tau)-Q_{0,\rho_0}(\tau)
\]
from the connected set $[0,1]^2$ into $\Rr$ is continuous.  Its image is connected and contains the sharp lower and upper QTE endpoints at the appropriate corners.  A connected subset of $\Rr$ containing these endpoints contains the whole interval between them.  Since \Cref{cor:qte} shows that no attainable scalar QTE value can lie outside this interval, the exact scalar QTE set equals the sharp interval hull.  Thus \Cref{ass:connected} holds at $(s,\tau)$.
\end{proof}

\begin{proof}[Proof of \Cref{cor:exactconnected}]
By \Cref{cor:qte}, the exact set $\mathfrak I_\Delta(s,\tau)$ contains the two endpoints $\Delta^-_s(\tau)$ and $\Delta^+_s(\tau)$ and is contained in their closed interval.  Under \Cref{ass:connected}, $\mathfrak I_\Delta(s,\tau)$ is connected.  A connected subset of $\Rr$ that contains two points contains the entire interval between them.  Therefore $\mathfrak I_\Delta(s,\tau)=[\Delta^-_s(\tau),\Delta^+_s(\tau)]$, and the equivalence for zero follows.
\end{proof}

\begin{proof}[Proof of \Cref{cor:reductions}]
If $\Gam=1$, then $\ell_\Gam(e)=u_\Gam(e)=1$, so $C^-_\Gam(p,e)=C^+_\Gam(p,e)=p$.  If $\Lam=1$, then $T^-_\Lam(q)=T^+_\Lam(q)=q$.  Combining these identities proves all claims.
\end{proof}

\section{Proofs for semiparametric theory}\label{app:semiparametric}

\begin{proof}[Proof of \Cref{lem:tangent}]
Along a regular path, write
\[
        \psi^\sigma_{a,s,t}(y)=\int G^\sigma_s\{p_{a,t}(y,x),e_{a,t}(x)\}\,dP^X_{0,t}(x).
\]
The derivative of the target covariate law is
\[
        \int b^\sigma_{a,s}(y,x)S^X_0(x)\,dP^X_0(x)
        =E\left[\frac{\ind(R=0)}{\pi_0}\{b^\sigma_{a,s}(y,X)-\psi^\sigma_{a,s}(y)\}S(O)\right],
\]
where the centering appears because $E(S^X_0\mid R=0)=0$.  For the source conditional outcome law,
\[
        \dot p_a(y,x)=E\!\left[\frac{\ind(A=a)}{e_a(X)}\{\ind(Y\le y)-p_a(y,X)\}S(O)\mid R=1,X=x\right].
\]
For the source treatment mechanism in the observational model,
\[
        \dot e_a(x)=E\!\left[\{\ind(A=a)-e_a(X)\}S(O)\mid R=1,X=x\right],
\]
whereas $\dot e_a=0$ in the known-design model.  Multiplying these conditional derivatives by $G_p^\sigma$ and $G_e^\sigma$, integrating over $P^X_0$, and using $dP^X_0=\omega\,dP^X_1$ gives the source term in the lemma.  Scores for the source covariate marginal law have zero contribution because the conditional residual $\zeta^\sigma_{a,s,y}$ has mean zero given $(R,X)=(1,x)$; scores for the sampling probability have zero contribution because both displayed components have conditional mean zero within their source strata.  This proves the decomposition.
\end{proof}

\begin{proof}[Proof of \Cref{thm:eif}]
Work along a regular parametric submodel with score $S(O)$ at the true law.  The target part of the functional is $E_0 b^\sigma_{a,s}(y,X)$.  Differentiating the target covariate law gives
\[
        \left.\frac{d}{dt}E_{0,t} b^\sigma_{a,s}(y,X)\right|_{0}
        =E\left[\frac{\ind(R=0)}{\pi_0}\{b^\sigma_{a,s}(y,X)-\psi^\sigma_{a,s}(y)\}S(O)\right].
\]
Now condition on $(R,X)=(1,x)$.  The derivative of $p_a(y,x)$ along the source conditional outcome law equals
\[
        E\left[\frac{\ind(A=a)}{e_a(X)}\{\ind(Y\le y)-p_a(y,X)\}S(O)\mid R=1,X=x\right].
\]
When $e_a$ is unknown, the derivative of $e_a(x)$ equals
\[
        E\left[\{\ind(A=a)-e_a(X)\}S(O)\mid R=1,X=x\right],
\]
and when $e_a$ is known by design this derivative is zero.  Since $G_s^\sigma$ is differentiable on the active region, the chain rule gives the conditional derivative
\[
        E\{\zeta^\sigma_{a,s,y}(O)S(O)\mid R=1,X=x\}.
\]
Averaging this derivative over the target covariate law is the same as averaging over $P^X_1$ with weight $\omega(x)$, which yields
\[
        E\left[\frac{\ind(R=1)}{\pi_1}\omega(X)\zeta^\sigma_{a,s,y}(O)S(O)\right].
\]
The sum of the target and source displays is $E\{\phi^\sigma_{a,s,y}(O)S(O)\}$, proving that the displayed function is an influence function.  The two conditional residuals used above are the canonical gradients for the conditional Bernoulli mean $p_a$ and the source propensity $e_a$ in the nonparametric source model; the target term is the canonical gradient for an expectation under $P^X_0$.  These tangent spaces are orthogonal under the mixture sampling model, so the displayed influence function is the canonical gradient and its variance is the efficiency bound.
\end{proof}

\begin{proof}[Proof of \Cref{prop:orthogonal}]
At the true nuisance value, the expectation of the score equals zero.  Let $\nu_t=(p_t,e_t,\omega_t)$ be a regular nuisance path through $\nu^\star$, and write $\dot p$, $\dot e$ and $\dot\omega$ for its derivative at $t=0$.  In the known-design model, $e_t=e$ and $\dot e=0$.  Work on a fixed active branch, so ordinary derivatives $G_p$ and $G_e$ exist.

The derivative of the target plug-in part is
\[
        \left.\frac{d}{dt}E_0G\{p_t(y,X),e_t(X)\}\right|_{t=0}
        =E_0\{G_p\dot p+G_e\dot e\},
\]
with the $G_e\dot e$ term absent in the known-design model.  The source augmentation has conditional expectation, under the true law,
\[
\begin{aligned}
A(t)&=E_1\Big[\omega_t(X)\Big\{
G_p\{p_t,e_t\}\frac{e}{e_t}(p-p_t)
+\chi G_e\{p_t,e_t\}(e-e_t)\Big\}\Big].
\end{aligned}
\]
At $t=0$ the expression in braces is zero, so the derivative of $\omega_t$ contributes nothing.  Differentiating the remaining factors and using the fact that derivatives of $G_p$, $G_e$ and $e/e_t$ multiply the zero residuals $p-p_t$ and $e-e_t$ at $t=0$, we obtain
\[
        \left.\frac{d}{dt}A(t)\right|_{t=0}
        =-E_1\{\omega(X)(G_p\dot p+\chi G_e\dot e)\}.
\]
Since $dP^X_0=\omega dP^X_1$, this equals $-E_0(G_p\dot p+G_e\dot e)$ in the observational-source model and $-E_0G_p\dot p$ in the known-design model.  It cancels the derivative of the target plug-in part.  Thus the pathwise derivative of the expected score with respect to every nuisance component is zero.  The calculation is local to a fixed active branch; at active-set ties the map is directional rather than linearly differentiable, and this proposition is not invoked.
\end{proof}

\begin{proof}[Proof of \Cref{prop:firststageprimitive}]
Condition on the training folds.  The Donsker and $L_2(P)$-continuity conditions imply stochastic equicontinuity of the empirical process indexed by the score class, so
\[
        \sup_{i\in\mathcal I_{\mathrm{reg}}}
        \left|\mathbb G_n\{m_i(\cdot;\psi_i,\widehat\nu)-m_i(\cdot;\psi_i,\nu^\star)\}\right|=o_p(1),
\]
where $i=(a,\sigma,y,s)$.  Cross-fitting permits this conditional argument without own-observation bias.
It remains to verify the deterministic drift.  Fix an index and abbreviate
$G=G^\sigma_s$, $p=p_a(y,X)$, $e=e_a(X)$, $\widehat p=\widehat p_a(y,X)$,
$\widehat e=\widehat e_a(X)$ and $\widehat\omega=\widehat\omega(X)$; set
$\widehat e=e$ when $\chi=0$.  On a fixed active branch, $G$ is of the form
$G(p,e)=c_0+c_1p+c_2e+c_3pe$ with coefficients uniformly bounded over
$\mathcal I_{\mathrm{reg}}$.  Taking conditional expectation of the cross-fitted score given the training sample gives
\[
\begin{aligned}
D_i(\widehat\nu)
&=E_1\Big[\omega\{G(\widehat p,\widehat e)-G(p,e)
  +G_p(\widehat p,\widehat e)\frac{e}{\widehat e}(p-\widehat p)
  +\chi G_e(\widehat p,\widehat e)(e-\widehat e)\}\Big]\\
&\quad +E_1\Big[(\widehat\omega-\omega)
  \{G_p(\widehat p,\widehat e)\frac{e}{\widehat e}(p-\widehat p)
  +\chi G_e(\widehat p,\widehat e)(e-\widehat e)\}\Big].
\end{aligned}
\]
The first line is the target plug-in drift plus the source augmentation drift with the true covariate ratio.  Since $G$ is bilinear on the active branch and $e/\widehat e$ is uniformly bounded by overlap, this first line is bounded in absolute value by a constant times
$E_1|\widehat p-p|\,|\widehat e-e|$; it is zero when $\chi=0$ because then $\widehat e=e$.  The second line is bounded by a constant times
\[
        E_1|\widehat\omega-\omega|\{ |\widehat p-p|+\chi |\widehat e-e|\}.
\]
Cauchy--Schwarz therefore gives the uniform bound
\[
        \sup_{i\in\mathcal I_{\mathrm{reg}}} |D_i(\widehat\nu)|
        \le C\{\delta_{p,n}\delta_{\omega,n}
        +\chi\delta_{e,n}\delta_{\omega,n}
        +\chi\delta_{p,n}\delta_{e,n}\}+o_p(n^{-1/2}),
\]
where the last $o_p(n^{-1/2})$ accounts for active-branch agreement and uniformly negligible approximation terms.  The product-rate condition in \Cref{ass:donsker} makes this drift $o_p(n^{-1/2})$.  Active-set agreement follows by assumption.  Finally, a Donsker class with square-integrable envelope is pre-Gaussian and satisfies the functional central limit theorem in $\ell^\infty(\mathcal I_{\mathrm{reg}})$.  These are precisely the components required in \Cref{ass:firststage}.
\end{proof}

\begin{proof}[Proof of \Cref{thm:ual}]
Write $P_n$ for the empirical measure and $\phi_i$ for the influence function indexed by $i=(a,\sigma,y,s)$.  Expanding the cross-fitted one-step estimator around the true nuisance functions gives, uniformly in $i\in\mathcal I_{\mathrm{reg}}$,
\[
        \widehat\psi_i-\psi_i=(P_n-P)\phi_i + R_{1n,i}+R_{2n,i},
\]
where $R_{1n,i}$ is the empirical-process error from replacing $\nu^\star$ by $\widehat\nu$ in the score and $R_{2n,i}$ is the deterministic drift of the orthogonal score.  Cross-fitting removes own-observation bias, \Cref{prop:orthogonal} eliminates the first-order nuisance drift, and \Cref{ass:firststage} gives
\[
        \sup_i\sqrt n|R_{1n,i}+R_{2n,i}|=o_p(1).
\]
The difference between using $n_r^{-1}\sum_{R_i=r}$ and $n^{-1}\sum_i \pi_r^{-1}\ind(R_i=r)$ is absorbed in the same influence function by the standard ratio expansion for sample splitting indicators.  The asserted asymptotic linearity follows.  The Gaussian approximation is the pre-Gaussian/Donsker or high-dimensional approximation condition in \Cref{ass:firststage}, applied to the centered influence-function class.
\end{proof}

\begin{proof}[Proof of \Cref{prop:directionalfinite}]
With finite covariate support, the estimators of $P^X_0$, $P^X_1$ and $e_a(x)$ are finite-dimensional smooth functions of multinomial cell proportions on events whose probabilities tend to one.  For each fixed $(a,x)$, the empirical source-arm CDF process
$y\mapsto\widehat p_a(y,x)-p_a(y,x)$ is indexed by half-lines and is therefore Donsker in $\ell^\infty(\calY)$; with finitely many arms and covariate values, the joint primitive process converges weakly in the stated product space.  Bounded cell probabilities make the map from $(P^X_0,P^X_1)$ to $\omega=dP^X_0/dP^X_1$ continuously differentiable.

The maps $C^\sigma_\Gam$, $T^\sigma_\Lam$ and $G^\sigma_s=T^\sigma_\Lam\circ C^\sigma_\Gam$ are finite maxima or minima of affine functions.  They are Hadamard directionally differentiable uniformly over compact $\calS$, with derivatives displayed in \Cref{thm:hdd}.  Because the target expectation is a finite sum over $x$, composition with the weakly convergent primitive process yields \Cref{ass:directional} on the full index set $\mathcal I$.  Standard recomputed $m$-out-of-$n$ subsampling for weakly convergent statistics, together with $m\to\infty$, $m/n\to0$, and continuity of the limiting critical-value distribution, gives the subsampling assertion.  Restricting $y$ and $s$ to a finite grid yields the finite-dimensional special case.
\end{proof}

\begin{proof}[Proof of \Cref{thm:hdd}]
The maps $(p,e)\mapsto g_j^\sigma(p,e)$ and $q\mapsto t_j^\sigma(q)$ are continuous affine maps.  The maps taking a finite maximum or minimum of affine maps are Hadamard directionally differentiable in sup norm, with derivative equal to the maximum or minimum of the active directional derivatives.  This gives the displayed derivatives $\dot C$ and $\dot T$.  Composition of Hadamard directionally differentiable maps gives $\dot G^\sigma$ by the directional chain rule.  Finally,
\[
        \psi^\sigma_{a,s}(y)=\int G^\sigma_s\{p_a(y,x),e_a(x)\}\,dP^X_0(x),
\]
so differentiating the integral gives the signed-measure direction $h_0\{b^\sigma_{a,s}(y,\cdot)\}$ plus the integral of the derivative of $G^\sigma_s$.  Boundedness of the outcome support, overlap, and \Cref{ass:directional} justify taking the derivative uniformly in $\ell^\infty(\mathcal I)$.  The final assertion is the directional delta method applied to the primitive weak limit.
\end{proof}

\begin{proof}[Proof of \Cref{cor:subsampling}]
By \Cref{thm:hdd} and \Cref{ass:directional}, the full-sample estimator and the recomputed subsample estimator are asymptotically obtained by applying the same Hadamard directionally differentiable map to primitive estimators with the same tight limit law, after scaling by $\sqrt n$ and $\sqrt m$, respectively.  The condition $m\to\infty$ and $m/n\to0$ makes the full-sample centering error negligible relative to the subsample scale.  The continuity of the limiting sup-norm distribution at the critical value is the standard condition ensuring quantile consistency for subsampling.  Therefore the conditional subsampling distribution of the displayed statistic consistently estimates the law of the directional-delta-method limit.  This is the usual subsampling argument for directionally differentiable statistics, applied to the CDF-bound process.
\end{proof}

\section{Proofs for quantile and frontier inference}\label{app:inference}

\begin{proof}[Proof of \Cref{thm:quantileband}]
On the simultaneous CDF-band event, the raw bands contain each CDF $\psi^\sigma_{a,s}$.  The monotone outer-envelope construction preserves containment: if $z\le y$, then $\widehat L^{raw}(z)\le \psi(z)\le\psi(y)$, so $\sup_{z\le y}\widehat L^{raw}(z)\le\psi(y)$; similarly, if $z\ge y$, then $\psi(y)\le\psi(z)\le\widehat U^{raw}(z)$, so $\psi(y)\le\inf_{z\ge y}\widehat U^{raw}(z)$.  Thus $\widehat L\le\psi\le\widehat U$ pointwise and the monotone envelopes are valid CDF bands.

If $L\le F\le U$, then $\{y:U(y)\ge\tau\}\supseteq\{y:F(y)\ge\tau\}\supseteq\{y:L(y)\ge\tau\}$.  Taking infima gives
\[
        \mathcal Q_L(\tau;L,U)\le F^{-1}(\tau)\le \mathcal Q_U(\tau;L,U).
\]
Apply this relation to $F=\psi^+_{a,s}$ for $q^-_{a,s}$ and to $F=\psi^-_{a,s}$ for $q^+_{a,s}$, uniformly over $(a,s,\tau)$ on the CDF-band event.
\end{proof}

\begin{proof}[Proof of \Cref{cor:qteband}]
On the event in \Cref{thm:quantileband},
\[
        q^-_{1,s}(\tau)\ge \widehat q^{-,lo}_{1,s}(\tau),
        \qquad
        q^+_{0,s}(\tau)\le \widehat q^{+,hi}_{0,s}(\tau),
\]
so
\[
        \Delta^-_s(\tau)=q^-_{1,s}(\tau)-q^+_{0,s}(\tau)
        \ge \widehat q^{-,lo}_{1,s}(\tau)-\widehat q^{+,hi}_{0,s}(\tau)
        =\widehat\Delta^{lo}_s(\tau).
\]
Similarly,
\[
        \Delta^+_s(\tau)=q^+_{1,s}(\tau)-q^-_{0,s}(\tau)
        \le \widehat q^{+,hi}_{1,s}(\tau)-\widehat q^{-,lo}_{0,s}(\tau)
        =\widehat\Delta^{hi}_s(\tau).
\]
This proves inclusion of the true identified interval in the estimated outer interval uniformly.
\end{proof}

\begin{proof}[Proof of \Cref{cor:qif}]
Under \Cref{ass:smoothq}, the inverse map is Hadamard differentiable at the relevant CDF.  If $q=F^{-1}(\tau)$ and $f(q)=F'(q)>0$, its derivative in direction $h$ is $-h(q)/f(q)$.  Apply this to $F=\psi^+_{a,s}$ for $q^-_{a,s}$ and to $F=\psi^-_{a,s}$ for $q^+_{a,s}$, using the CDF influence functions from \Cref{thm:eif}.  The QTE endpoint formulas follow by linearity.
\end{proof}

\begin{proof}[Proof of \Cref{prop:kappa}]
By \Cref{cor:qif}, each arm-specific quantile endpoint has a uniformly linear expansion under active-set and smooth-quantile regularity.  The maps defining $\Delta^-_s(\tau)$ and $\Delta^+_s(\tau)$ are linear combinations of these quantile endpoints, so their influence functions are $\varphi^-_{s,\tau}$ and $\varphi^+_{s,\tau}$.  If the active argument of the minimum in $\kappa^{\mathrm{hull}}_s(\tau)=\min\{\Delta^+_s(\tau),-\Delta^-_s(\tau)\}$ is separated uniformly, the minimum map is locally linear and selects the displayed branch.  Uniformity follows from the assumed uniformity of the quantile expansions and the separation margin.  Conditional multiplier consistency follows from the same empirical-process approximation used in \Cref{thm:ual}.
\end{proof}

\begin{proof}[Proof of \Cref{thm:frontier}]
On the event
\[
        \sup_{s,\tau}\sqrt n|\widehat\kappa^{\mathrm{hull}}_s(\tau)-\kappa^{\mathrm{hull}}_s(\tau)|\le d_{1-\alpha},
\]
if $s\in\widehat{\mathcal N}^{in}_{\alpha}(\tau)$, then
\[
        \kappa^{\mathrm{hull}}_s(\tau)\ge \widehat\kappa^{\mathrm{hull}}_s(\tau)-d_{1-\alpha}/\sqrt n\ge0,
\]
so $s\in\mathcal N^{\mathrm{hull}}(\tau)$.  Conversely, if $s\in\mathcal N^{\mathrm{hull}}(\tau)$, then $\kappa^{\mathrm{hull}}_s(\tau)\ge0$ and
\[
        \widehat\kappa^{\mathrm{hull}}_s(\tau)+d_{1-\alpha}/\sqrt n\ge \kappa^{\mathrm{hull}}_s(\tau)\ge0,
\]
so $s\in\widehat{\mathcal N}^{out}_{\alpha}(\tau)$.  This proves the set inclusion uniformly on the event, whose probability tends to at least $1-\alpha$.

For the local level-set statement, let
\[
        \Delta_n=\sup_{s\in K,\tau\in\calT}
        |\widehat\kappa^{\mathrm{hull}}_s(\tau)-\kappa^{\mathrm{hull}}_s(\tau)|+d_{1-\alpha}/\sqrt n.
\]
The uniform band gives $\Delta_n=O_p(n^{-1/2})$.  The isolation condition implies that, with probability tending to one, the estimated zero-level equation
$\widehat\kappa^{\mathrm{hull}}_s(\tau)+d_{1-\alpha}/\sqrt n=0$ has no solutions in $K$ outside any fixed small tube around $\mathcal F^{\mathrm{hull}}_{\mathrm{int}}(\tau)$, uniformly in $\tau$.

Inside such a tube, compact containment and the lower bound on $\|\nabla_s\kappa^{\mathrm{hull}}\|$ allow a uniform implicit-function representation. If the estimator is evaluated only on a grid, interpolate it continuously on each grid cell; the additional Hausdorff error is bounded by the mesh size, which is $o(n^{-1/2})$ under the stated grid convention. In local coordinates $(u,v)$, the true frontier is $v=0$ and $\kappa^{\mathrm{hull}}(u,v,\tau)$ is strictly monotone in $v$, with derivative bounded away from zero uniformly.  Thus any solution of the estimated zero-level equation inside the tube lies at normal distance at most $C\Delta_n$ from the true frontier.  Conversely, for every point on the true frontier, moving along the normal coordinate by $\pm C\Delta_n$ changes the true $\kappa^{\mathrm{hull}}$ by more than the maximal perturbation; the intermediate value theorem therefore gives a solution of the estimated zero-level equation within distance $C\Delta_n$.  Hence the Hausdorff distance between $\widehat{\mathcal F}^{out}_{\alpha,K}(\tau)$ and $\mathcal F^{\mathrm{hull}}_{\mathrm{int}}(\tau)$ is $O_p(n^{-1/2})$ uniformly in $\tau$.  Under \Cref{ass:connected}, \Cref{cor:exactconnected} identifies the interval hull with the exact scalar QTE set, so the exact-frontier statement follows by substitution.
\end{proof}

\section{Algebraic audit and numerical stress checks}\label{app:audit}

This appendix records two internal checks used to guard against common errors in the nested-envelope algebra.  First, the proof of \Cref{prop:noncollapse} shows analytically that replacing the two normalized tilts by one product tilt is a relaxation: the product lower endpoint is weakly below the nested lower endpoint and the product upper endpoint is weakly above the nested upper endpoint.  Hence a paper using only the product likelihood-ratio bound would generally report conservative, not sharp, transported CDF bounds for the joint model in this manuscript.

Second, the closed forms were stress-tested against the finite binary-event linear program.  For randomly generated $p\in[0,1]$, $e\in(0,1)$, $\Gam\ge1$ and $\Lam\ge1$, the feasible interval for $m=E(hZ)$ under $\ell\le h\le u$ and $E(h)=1$ is exactly
\[
        [\ell p,up]\cap[1-u(1-p),1-\ell(1-p)].
\]
The numerical check repeatedly compared this linear-program endpoint with $C^-_\Gam$, $C^+_\Gam$, $T^-_\Lam$ and $T^+_\Lam$, and also verified monotonicity of the resulting CDF maps on random grids.  No violations were found up to floating-point error.  This check is not used as a proof; it is included to make transparent that the closed-form formulas have been compared with the underlying finite-dimensional optimization problem.

\newpage
\section{Additional simulation diagnostics}\label{app:simulation-diagnostics}

\begin{table}[!htbp]
\centering
\scriptsize
\setlength{\tabcolsep}{4pt}
\renewcommand{\arraystretch}{1.15}
\caption{Full diagnostics for Experiment 1.  The table reports the numerical audit of the finite-support sharpness formula, the non-collapse comparison between the nested model and the product-LR relaxation, and whole-path validity checks for least-favorable CDF constructions.}
\label{tab:sim1_finite_support_audit}
\begin{tabularx}{0.98\textwidth}{>{\raggedright\arraybackslash}p{0.30\textwidth}>{\raggedleft\arraybackslash}p{0.16\textwidth}>{\raggedright\arraybackslash}X}
\toprule
Metric & Value & Definition \\
\midrule
Full finite LP solves & $6{,}000$ & Number of threshold-level two-layer LPs solved using \texttt{scipy.optimize.linprog}. \\
Vectorized algebraic cases & $200{,}000$ & Independent event-mass cases used for non-collapse and dominance diagnostics. \\
Process-path audits & $20{,}000$ & Whole-CDF lower/upper path constructions checked for monotonicity, endpoint conditions, ordering, and admissible Radon--Nikodym derivatives. \\
Max LP discrepancy, lower & $4.955\times 10^{-11}$ & $\max |T^-_{\Lam}\{C^-_{\Gam}(p,e)\}-v^-_{\mathrm{LP}}|$. \\
Max LP discrepancy, upper & $3.595\times 10^{-11}$ & $\max |T^+_{\Lam}\{C^+_{\Gam}(p,e)\}-v^+_{\mathrm{LP}}|$. \\
Boundary-event max discrepancy & $7.105\times 10^{-15}$ & Maximum LP discrepancy among threshold events with $p=0$ or $p=1$. \\
Mean product overwidth & $0.002487$ & Average $W^{\mathrm{prod}}-W^{\mathrm{nest}}$; theory predicts this quantity is nonnegative. \\
Strict non-collapse share & $0.163000$ & Fraction of all algebraic cases with $W^{\mathrm{prod}}>W^{\mathrm{nest}}+10^{-10}$. \\
Strict non-collapse share, nontrivial subset & $0.258299$ & Same fraction restricted to $\Gam>1$, $\Lam>1$, interior event probability and positive nested width. \\
99\% product/nested width ratio, nontrivial subset & $1.113234$ & Upper-tail width ratio restricted to the nontrivial subset. \\
Analytic example: nested lower & $0.286667$ & $T^-_{\Lam}\{C^-_{\Gam}(p,e)\}$ at $p=0.7$, $e=0.1$, $\Gam=2$, and $\Lam=1.5$. \\
Analytic example: product lower & $0.256667$ & Single product-LR lower endpoint at the same $(p,e,\Gam,\Lam)$. \\
Product dominance violations & $0$ & Cases in which the product relaxation is narrower than the nested model. \\
Max process-path violation & $2.487\times 10^{-14}$ & Maximum numerical violation of monotonicity, endpoint, ordering or Radon--Nikodym constraints. \\
Elapsed seconds & $33.872$ & Wall-clock runtime for this run. \\
\bottomrule
\end{tabularx}
\end{table}

\begin{figure}[!htbp]
    \centering
    \includegraphics[width=0.98\textwidth]{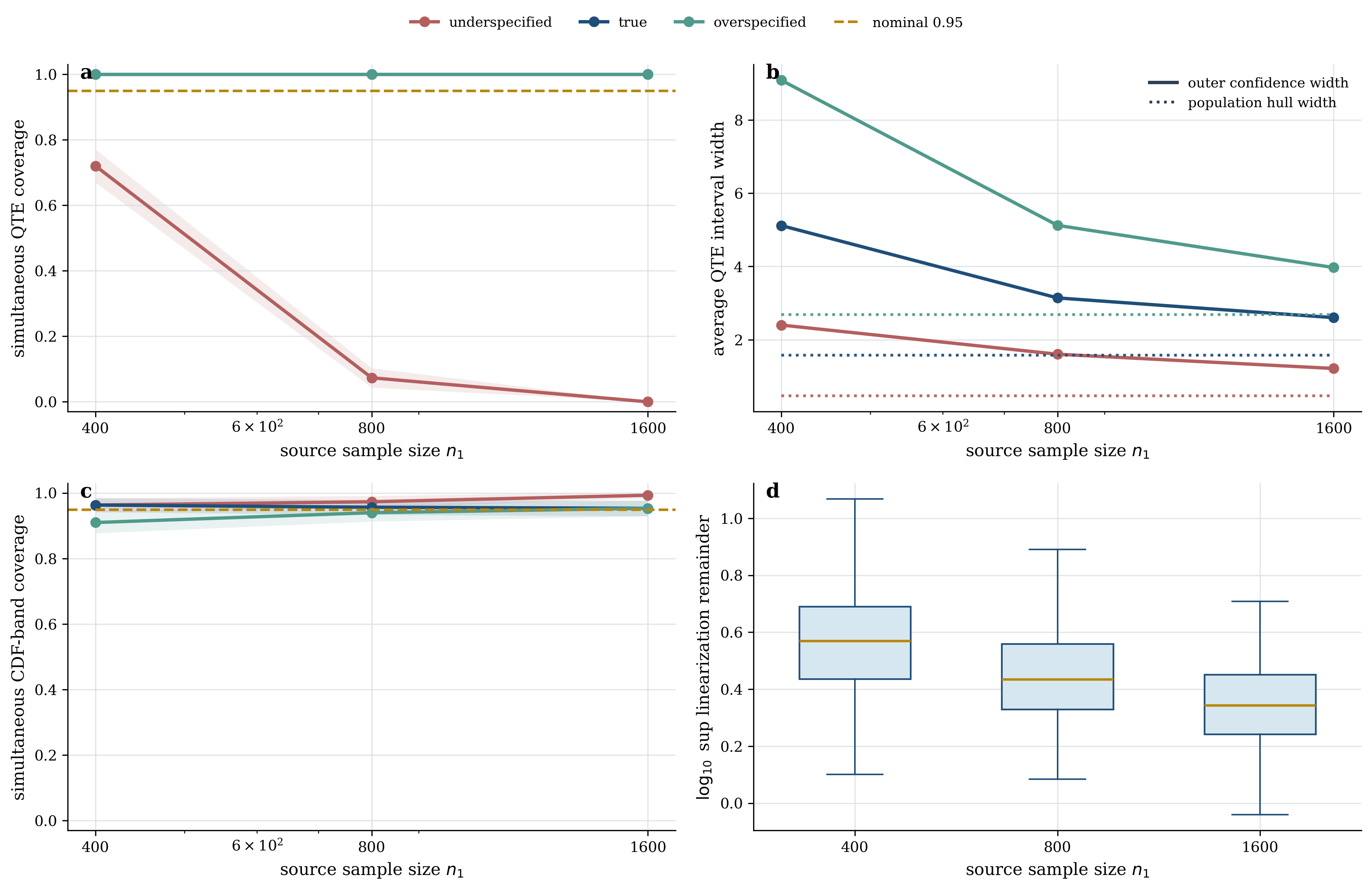}
    \caption{Regular smooth inference under exact nested tilts.  Panel (a) reports simultaneous QTE coverage over the quantile grid.  Panel (b) reports average QTE interval width; solid lines show outer confidence hulls and dotted lines show oracle population hulls.  Panel (c) reports simultaneous CDF-band coverage.  Panel (d) reports the first-order linearization diagnostic under the true sensitivity pair.  Full numerical diagnostics are reported in \Cref{tab:sim2_regular_appendix}.}
    \label{fig:sim2_regular_smooth}
\end{figure}

\begin{table}[!htbp]
\centering
\scriptsize
\setlength{\tabcolsep}{3pt}
\renewcommand{\arraystretch}{1.08}
\caption{Extended diagnostics for Experiment 2.  QTE coverage is simultaneous over the quantile grid; CDF coverage is simultaneous over treatment arms, endpoint signs and thresholds.  Plug-in coverage uses the estimated interval hull without CDF-band enlargement.}
\label{tab:sim2_regular_appendix}
\resizebox{\textwidth}{!}{%
\begin{tabular}{rrlrrrrrrrr}
\toprule
$n_1$ & $n_0$ & Sensitivity & QTE cov. & CDF cov. & plug-in cov. & outer width & plug-in width & pop. width & CDF width & lin. rem. \\
\midrule
$400$ & $600$ & overspecified & $1.000$ & $0.910$ & $1.000$ & $9.087$ & $2.988$ & $2.689$ & $0.390$ & -- \\
$400$ & $600$ & underspecified & $0.720$ & $0.963$ & $0.000$ & $2.398$ & $0.509$ & $0.465$ & $0.219$ & -- \\
$400$ & $600$ & true & $1.000$ & $0.963$ & $0.453$ & $5.116$ & $1.750$ & $1.579$ & $0.292$ & $3.713$ \\
$800$ & $1200$ & overspecified & $1.000$ & $0.940$ & $1.000$ & $5.122$ & $2.855$ & $2.689$ & $0.263$ & -- \\
$800$ & $1200$ & underspecified & $0.073$ & $0.973$ & $0.000$ & $1.602$ & $0.477$ & $0.465$ & $0.148$ & -- \\
$800$ & $1200$ & true & $1.000$ & $0.957$ & $0.323$ & $3.140$ & $1.651$ & $1.579$ & $0.196$ & $2.723$ \\
$1600$ & $2400$ & overspecified & $1.000$ & $0.953$ & $1.000$ & $3.970$ & $2.749$ & $2.689$ & $0.183$ & -- \\
$1600$ & $2400$ & underspecified & $0.000$ & $0.993$ & $0.000$ & $1.213$ & $0.470$ & $0.465$ & $0.103$ & -- \\
$1600$ & $2400$ & true & $1.000$ & $0.953$ & $0.307$ & $2.602$ & $1.603$ & $1.579$ & $0.137$ & $2.205$ \\
\bottomrule
\end{tabular}%
}
\end{table}

\begin{table}[!htbp]
\centering
\scriptsize
\setlength{\tabcolsep}{3pt}
\renewcommand{\arraystretch}{1.08}
\caption{Extended diagnostics for Experiment 3.  QTE coverage is simultaneous over the quantile grid.  The process sup-norm error is computed against the oracle bound process for sensitivity-aware estimators and against the true target CDF for the point estimator.  The $G_e$ diagnostics are reported only for the full DML estimator because the other methods deliberately omit this EIF component.}
\label{tab:sim3_hdml_appendix}
\resizebox{\textwidth}{!}{%
\begin{tabular}{rrlrrrrrrrrrr}
\toprule
$n_1$ & $n_0$ & Method & QTE cov. & Hull cont. & RMSE & Abs. bias & outer width & hull width & Proc. err. & nuis. CDF ISE & $G_e$ act. & $|G_e|$ \\
\midrule
$600$ & $900$ & DML & $1.000$ & $0.620$ & $0.626$ & $0.531$ & $15.699$ & $3.210$ & $0.517$ & $0.0180$ & $0.379$ & $0.095$ \\
$600$ & $900$ & Plug-in & $0.120$ & $0.040$ & $0.565$ & $0.483$ & $2.028$ & $1.646$ & $0.399$ & $0.0180$ & \multicolumn{1}{c}{--} & \multicolumn{1}{c}{--} \\
$600$ & $900$ & No-$G_e$ & $1.000$ & $0.655$ & $0.648$ & $0.552$ & $15.671$ & $3.298$ & $0.569$ & $0.0180$ & \multicolumn{1}{c}{--} & \multicolumn{1}{c}{--} \\
$600$ & $900$ & Point & $0.850$ & $0.000$ & $1.230$ & $1.202$ & $6.795$ & $0.000$ & $0.525$ & $0.0180$ & \multicolumn{1}{c}{--} & \multicolumn{1}{c}{--} \\
$1200$ & $1800$ & DML & $1.000$ & $0.480$ & $0.392$ & $0.329$ & $12.509$ & $2.869$ & $0.337$ & $0.0081$ & $0.472$ & $0.110$ \\
$1200$ & $1800$ & Plug-in & $0.140$ & $0.040$ & $0.364$ & $0.303$ & $2.229$ & $2.011$ & $0.326$ & $0.0081$ & \multicolumn{1}{c}{--} & \multicolumn{1}{c}{--} \\
$1200$ & $1800$ & No-$G_e$ & $1.000$ & $0.525$ & $0.454$ & $0.384$ & $12.581$ & $3.011$ & $0.402$ & $0.0081$ & \multicolumn{1}{c}{--} & \multicolumn{1}{c}{--} \\
$1200$ & $1800$ & Point & $0.075$ & $0.000$ & $1.231$ & $1.218$ & $2.324$ & $0.000$ & $0.541$ & $0.0081$ & \multicolumn{1}{c}{--} & \multicolumn{1}{c}{--} \\
$2400$ & $3600$ & DML & $1.000$ & $0.470$ & $0.232$ & $0.188$ & $8.335$ & $2.626$ & $0.218$ & $0.0038$ & $0.510$ & $0.117$ \\
$2400$ & $3600$ & Plug-in & $0.100$ & $0.065$ & $0.249$ & $0.200$ & $2.343$ & $2.211$ & $0.266$ & $0.0038$ & \multicolumn{1}{c}{--} & \multicolumn{1}{c}{--} \\
$2400$ & $3600$ & No-$G_e$ & $1.000$ & $0.390$ & $0.286$ & $0.233$ & $8.507$ & $2.717$ & $0.280$ & $0.0038$ & \multicolumn{1}{c}{--} & \multicolumn{1}{c}{--} \\
$2400$ & $3600$ & Point & $0.000$ & $0.000$ & $1.247$ & $1.240$ & $1.177$ & $0.000$ & $0.548$ & $0.0038$ & \multicolumn{1}{c}{--} & \multicolumn{1}{c}{--} \\
\bottomrule
\end{tabular}%
}
\end{table}

\begin{table}[!htbp]
\centering
\scriptsize
\setlength{\tabcolsep}{3pt}
\renewcommand{\arraystretch}{1.08}
\caption{Extended diagnostics for Experiment 4.  QTE coverage is simultaneous over the quantile grid.  Frontier outer coverage and inner validity are evaluated on the compact $(\Gam,\Lam)$ grid.  Hull containment, plug-in endpoint RMSE and plug-in tipping error are diagnostics of the common estimated CDF-bound process; the methods differ in how they construct outer QTE and frontier confidence sets.}
\label{tab:sim4_nonregular_appendix}
\resizebox{\textwidth}{!}{%
\begin{tabular}{rrlrrrrrrrr}
\toprule
$n_1$ & $n_0$ & Method & QTE cov. & Hull cont. & QTE width & Plug-in RMSE & Front. outer & Inner valid & Hausdorff & Plug-in tip err. \\
\midrule
$500$ & $750$ & Subsample $m=n^{0.6}$ & $1.000$ & $0.247$ & $1.488$ & $0.117$ & $0.930$ & $1.000$ & $0.501$ & $0.095$ \\
$500$ & $750$ & Subsample $m=n^{0.7}$ & $1.000$ & $0.247$ & $1.604$ & $0.117$ & $0.990$ & $1.000$ & $0.554$ & $0.095$ \\
$500$ & $750$ & Regular bootstrap & $1.000$ & $0.247$ & $1.875$ & $0.117$ & $1.000$ & $1.000$ & $0.693$ & $0.095$ \\
$500$ & $750$ & Wald endpoint & $0.550$ & $0.247$ & $0.773$ & $0.117$ & -- & -- & -- & -- \\
$1000$ & $1500$ & Subsample $m=n^{0.6}$ & $1.000$ & $0.267$ & $1.276$ & $0.079$ & $0.987$ & $1.000$ & $0.461$ & $0.056$ \\
$1000$ & $1500$ & Subsample $m=n^{0.7}$ & $1.000$ & $0.267$ & $1.371$ & $0.079$ & $0.993$ & $1.000$ & $0.485$ & $0.056$ \\
$1000$ & $1500$ & Regular bootstrap & $1.000$ & $0.267$ & $1.461$ & $0.079$ & $1.000$ & $1.000$ & $0.520$ & $0.056$ \\
$1000$ & $1500$ & Wald endpoint & $0.603$ & $0.267$ & $0.721$ & $0.079$ & -- & -- & -- & -- \\
$2000$ & $3000$ & Subsample $m=n^{0.6}$ & $1.000$ & $0.273$ & $1.119$ & $0.061$ & $0.977$ & $1.000$ & $0.384$ & $0.046$ \\
$2000$ & $3000$ & Subsample $m=n^{0.7}$ & $1.000$ & $0.273$ & $1.172$ & $0.061$ & $0.987$ & $1.000$ & $0.404$ & $0.046$ \\
$2000$ & $3000$ & Regular bootstrap & $1.000$ & $0.273$ & $1.186$ & $0.061$ & $0.993$ & $1.000$ & $0.441$ & $0.046$ \\
$2000$ & $3000$ & Wald endpoint & $0.567$ & $0.273$ & $0.681$ & $0.061$ & -- & -- & -- & -- \\
\bottomrule
\end{tabular}%
}
\end{table}

\section{Additional real-data details}\label{app:realdata}

This appendix gives additional details for the real-data illustration in \Cref{sec:realdata}. 
We use the complete-case NHEFS data distributed with the \texttt{causaldata} package, which contains the smoking-cessation example from \citet{HernanRobins2020}; the underlying public-use NHEFS files are provided by CDC/NCHS \citep{CDCNHEFS,CausalData2024}. 
The complete-case version keeps subjects with complete data on the smoking-cessation treatment, baseline covariates, baseline weight, follow-up weight, and weight-change variables.

The analysis uses
\[
        A=\texttt{qsmk},\qquad
        Y=\texttt{wt82\_71}.
\]
Here \(A=1\) denotes quitting smoking, and \(Y\) is weight change from 1971 to 1982. 
The baseline covariate vector \(X\) contains only pre-outcome variables. 
In the implementation, the candidate baseline covariates are
\[
\begin{aligned}
        &\texttt{sex},\ \texttt{race},\ \texttt{income},\ \texttt{marital},\ \texttt{education},\ \texttt{school},\ \texttt{age},\ \texttt{wt71},\\
        &\texttt{smokeintensity},\ \texttt{smokeyrs},\ \texttt{exercise},\ \texttt{active},\ \texttt{hbp},\ \texttt{asthma},\ \texttt{bronch},\\
        &\texttt{diabetes},\ \texttt{alcoholfreq},\ \texttt{alcoholpy},\ \texttt{price71},\ \texttt{tax71},
\end{aligned}
\]
with variables included whenever available in the downloaded data. 
Post-outcome or post-treatment variables, including follow-up weight and post-baseline smoking changes, are not used as covariates.

The public data include an \texttt{older} indicator. 
We use \(\texttt{older}=0\) as the source sample and \(\texttt{older}=1\) as the target covariate sample. 
This produces \(n_1=1009\) source observations and \(n_0=413\) target observations. 
For source observations, \(A,Y,X\) are used. 
For target observations, only \(X\) is used; the observed target treatment and outcome variables in the public data are masked before estimation. 
The target treatment rate reported in \Cref{tab:realdata_nhefs_appendix} is therefore only a design diagnostic and does not enter any estimator.

The nuisance functions are estimated with five-fold cross-fitting. 
The treatment mechanism \(e_a(x)=P(A=a\mid R=1,X=x)\) is estimated in the source sample. 
The covariate density ratio \(\omega(x)=dP^X_0/dP^X_1\) is estimated by fitting a source--target sampling-score classifier on pooled covariates and converting the fitted odds into a density-ratio estimate. 
The conditional CDFs \(p_a(y,x)=P(Y\le y\mid R=1,A=a,X=x)\) are estimated separately by treatment arm on a threshold grid of size 181. 
The main table uses 399 multiplier draws for confidence-band diagnostics. 
For the descriptive robustness frontier, we evaluate a \(51\times51\) grid over \(\Gam\in[1,4]\) and \(\Lam\in[1,3]\), and report the smallest product \(\Gam\Lam\) for which the plug-in median-QTE interval hull contains zero.

\begin{table}[!htbp]
\centering
\small
\setlength{\tabcolsep}{4pt}
\renewcommand{\arraystretch}{1.10}
\caption{
Additional diagnostics for the NHEFS real-data illustration. 
The target treatment and outcome columns are not used by any estimator; target treatment rate is shown only to document the public-data split. 
The descriptive frontier tip is computed on the plug-in \(51\times51\) sensitivity grid over \(\Gam\in[1,4]\) and \(\Lam\in[1,3]\).
}
\label{tab:realdata_nhefs_appendix}
\begin{tabularx}{0.98\textwidth}{>{\raggedright\arraybackslash}p{0.28\textwidth}>{\raggedleft\arraybackslash}p{0.16\textwidth}>{\raggedright\arraybackslash}X}
\toprule
Diagnostic & Value & Definition \\
\midrule
Source sample size 
& \(1009\) 
& Younger-smoker source sample with treatment and outcome observed. \\
Target covariate sample size 
& \(413\) 
& Older-smoker target sample; treatment and outcome masked. \\
Treatment rate, source 
& \(0.221\) 
& Fraction with \(\texttt{qsmk}=1\) in the source sample. \\
Treatment rate, target (masked) 
& \(0.322\) 
& Reported only as a design diagnostic; not used for estimation. \\
Max absolute SMD 
& \(3.063\) 
& Largest absolute standardized mean difference between target and source covariates. \\
Propensity AUC 
& \(0.555\) 
& Held-out AUC for the source treatment nuisance learner. \\
Sampling-score AUC 
& \(0.989\) 
& Held-out AUC for discriminating target versus source covariates. \\
Median estimated density ratio 
& \(0.033\) 
& Median estimated \(dP_0^X/dP_1^X\) among source observations. \\
95\% estimated density ratio 
& \(1.018\) 
& 95th percentile of estimated \(dP_0^X/dP_1^X\) among source observations. \\
Descriptive frontier tip 
& \(1.440\) 
& Smallest \(\Gam\Lam\) on the plug-in grid for which the median-QTE interval hull contains zero. \\
\bottomrule
\end{tabularx}
\end{table}

\paragraph{Interpretation of \Cref{tab:realdata_nhefs_appendix}.}
The diagnostics confirm that the real-data split creates a meaningful transport problem. 
The maximum absolute standardized mean difference is \(3.063\), and the sampling-score AUC is \(0.989\), so the target older-smoker covariate distribution is far from the source younger-smoker covariate distribution. 
The estimated density ratio is also highly concentrated among source observations: its median is \(0.033\), while its \(95\%\) percentile is \(1.018\). 
These values suggest that only a subset of source observations receives substantial target-representation weight. 
This is precisely the setting in which point transport is vulnerable to external-validity violations.

The propensity AUC is \(0.555\), indicating that quitting status is only weakly predictable from the source baseline covariates used by the learner. 
This does not rule out unmeasured confounding; rather, it motivates a sensitivity analysis because treatment assignment in the source remains observational. 
The descriptive frontier tip \(1.440\) indicates that, on the plug-in \((\Gam,\Lam)\) grid, a moderate product of the two sensitivity parameters is enough for the median-QTE interval hull to include zero. 
This number should not be interpreted as an estimate of the true sensitivity level. 
It is a robustness summary: it reports how far one must move away from the point-transport model, along the grid used here, before the median-QTE conclusion becomes non-refuting for the null.

\end{document}